\documentclass[11pt,twoside]{article}
\usepackage[margin=1in]{geometry}

\usepackage{jeffe,graphicx}	
\usepackage{microtype}
\usepackage[mathletters]{ucs}		
\usepackage[utf8x]{inputenc}
\usepackage[charter]{mathdesign}	
\usepackage{berasans, beramono}
\usepackage{eucal,textcomp}
\usepackage[nocompress]{cite}		
\usepackage{enumerate}
\usepackage{pgf,tikz}
\usepackage{wrapfig}
\usepackage[medium,compact]{titlesec}

\usepackage[mathlines, pagewise]{lineno}
\linenumbers

\urldef\paperURL\url{http://www.cs.illinois.edu/~jeffe/pubs/weak.html}

\newtheorem{theorem}{Theorem}[section]
\newtheorem{lemma}[theorem]{Lemma}
\newtheorem{corollary}[theorem]{Corollary}
\numberwithin{figure}{section}

\newcommand{\bdry}{\partial\!}

\begin{document}

\begin{titlepage}

\title{Detecting Weakly Simple Polygons\thanks{
Work on this paper was partially supported by NSF grant CCF-0915519.  See \paperURL\ for the most recent version of this paper.}
}
\author{
	Hsien-Chih Chang 
	\qquad\quad
	Jeff Erickson
	\qquad\quad
	Chao Xu
\\[1ex]
Department of Computer Science\\
University of Illinois, Urbana-Champaign
}

\date{Submitted to SODA 2015 — July 7, 2014}

\maketitle

\begin{bigabstract}
A closed curve in the plane is weakly simple if it is the limit (in the Fréchet metric) of a sequence of simple closed curves.  We describe an algorithm to determine whether a closed walk of length $n$ in a simple plane graph is weakly simple in $O(n\log n)$ time, improving an earlier $O(n^3)$-time algorithm of Cortese \etal\ [\emph{Discrete Math.}\ 2009].  As an immediate corollary, we obtain the first efficient algorithm to determine whether an arbitrary $n$-vertex polygon is weakly simple; our algorithm runs in $O(n^2\log n)$ time.  We also describe algorithms that detect weak simplicity in $O(n\log n)$ time for two interesting classes of polygons. Finally, we discuss subtle errors in several previously published definitions of weak simplicity.
\end{bigabstract}
%

\setcounter{page}{0}
\thispagestyle{empty}
\end{titlepage}

\pagestyle{myheadings}
\markboth{Hsien-Chih Chang, Jeff Erickson, and Chao Xu}{Detecting Weakly Simple Polygons}

\newpage
\section{Introduction}

Simple polygons in the plane have been standard objects of study in computational geometry for decades, and in the broader mathematical community for centuries.  
Many algorithms designed for simple polygons continue to work with little or no modification in degenerate cases, where intuitively the polygon overlaps itself but does not cross itself.  
We offer the first complete and efficient algorithm to detect such degenerate polygons.

\begin{wrapfigure}{r}{3.25in}
\vspace{-2ex}
\centering
\includegraphics[scale=0.1]{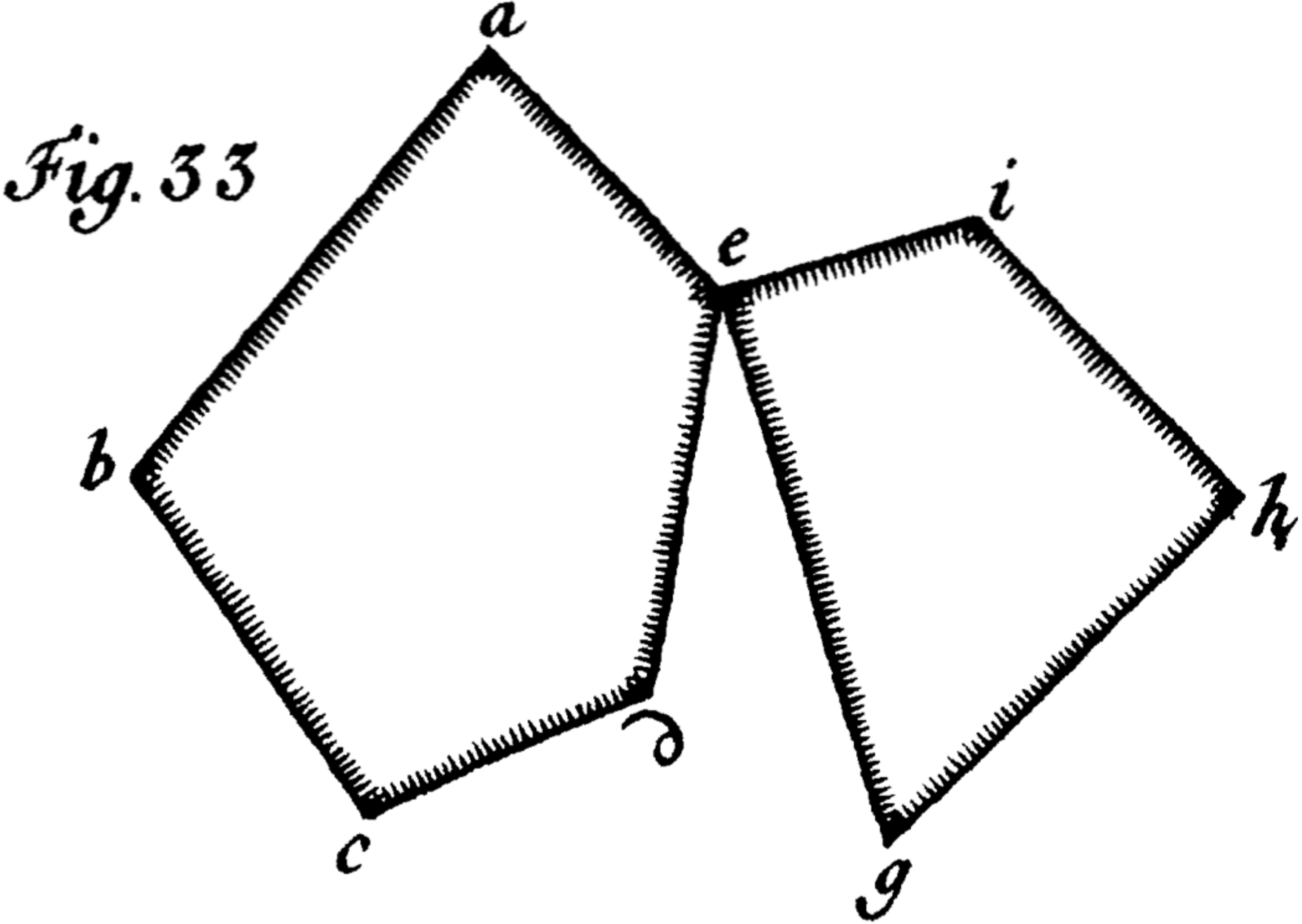}
\includegraphics[scale=0.1]{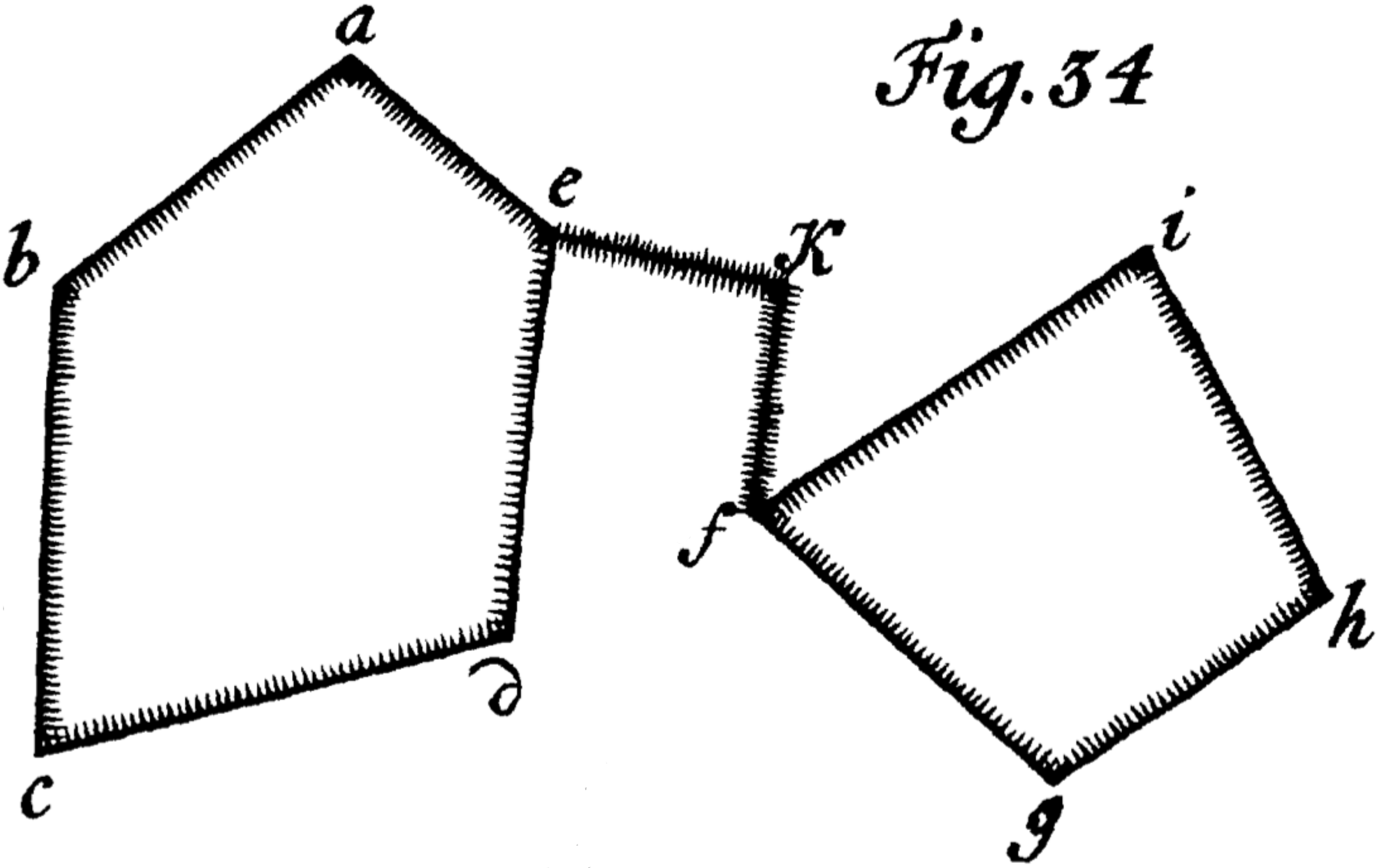}\\[1ex]
\includegraphics[scale=0.45]{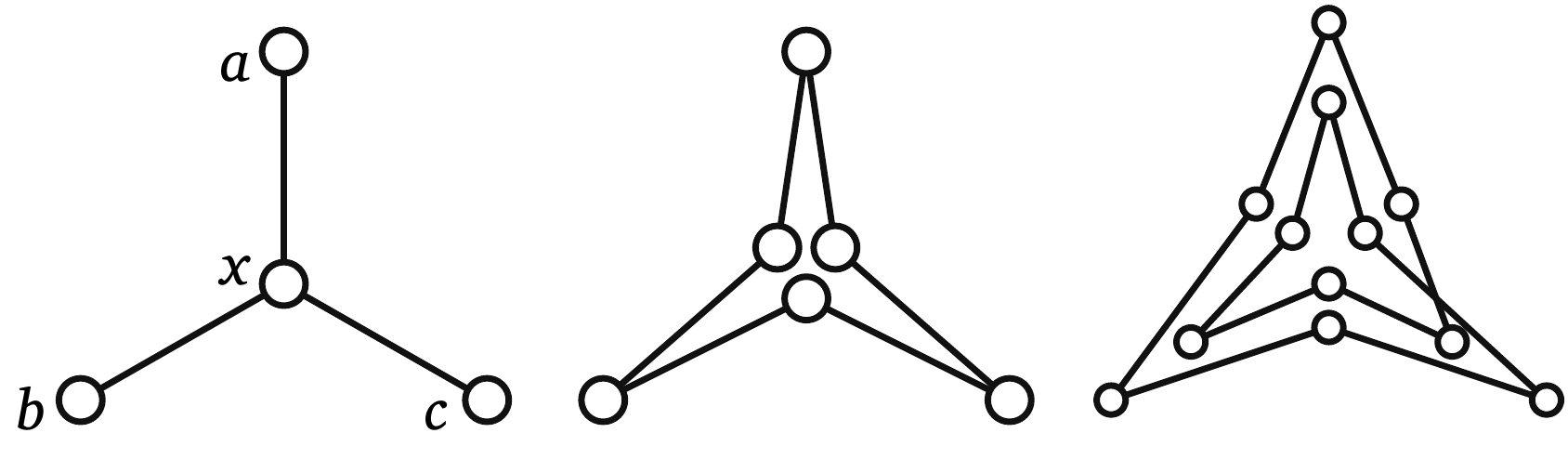}
\caption{Top row: Two weakly simple polygons from Meister’s seminal 1770 treatise on polygons \cite{m-ggfpi-1769}.	
Bottom row: The polygon $(a,x,b,x,c,x)$ is weakly simple; the polygon $(a,x,b,x,c,x,a,x,b,x,c,x)$ is not.
}
\label{F:first-example}
\vspace{-3ex}
\end{wrapfigure}

Formally, a closed curve in the plane is a continuous function $P\colon S^1\to\R^2$.  
A closed curve is \emph{simple} if it is injective, and \emph{weakly simple} if for any $\e>0$, there is a simple closed curve whose Fréchet distance from $P$ is at most~$\e$.  
A recent nontrivial result of Ribó Mor~\cite{r-rcpps-06} implies that a \emph{polygon}~$P$ with at least three vertices is weakly simple if and only if, for any $\e>0$, we can obtain a simple polygon by perturbing each vertex of~$P$ within a ball of radius~$\e$.  
See Figure~\ref{F:first-example} for some small examples.

Unfortunately, neither of these definitions imply efficient algorithms to determine whether a given polygon is weakly simple.  Several authors have offered alternative characterizations of weakly simple polygons, all building on the intuition that that a weakly simple polygon does not “cross itself”.  Unfortunately, \emph{none} of these characterizations is entirely consistent with the formal definition, especially when the algorithm contains \emph{spurs}---vertices whose two incident edges overlap.  We discuss these earlier definitions in detail in Section \ref{S:earlier} and Appendix \ref{S:earlier2}.

We describe an algorithm to determine whether a closed walk of length $n$ in a planar straight-line graph is weakly simple in $O(n\log n)$ time.  Our algorithm is essentially a more efficient implementation of an earlier $O(n^3)$-time algorithm of Cortese \etal~\cite{cbpp-ecpg-09} for the same problem.  Since any $n$-vertex polygon can be decomposed into a walk of length $O(n^2)$ in union of the polygon’s vertices and edges, we immediately obtain an algorithm to decide whether a polygon is weakly simple in $O(n^2\log n)$ time.  The quadratic blowup is caused by \emph{forks} in the polygon: vertices that lie in the interior of many overlapping collinear edges.  For polygons without forks, the running time is simply $O(n\log n)$.  We also describe a simpler $O(n\log n)$-time algorithm for polygons without spurs (but possibly with forks).

Our paper is organized as follows.  We review standard terminology, formally define weak simplicity, and discuss connections to other published definitions in Section \ref{S:back}.  In Section \ref{S:nodes}, as a warm-up to our later results, we describe a simple algorithm to determine whether a polygon without spurs or forks is weakly simple in $O(n\log n)$ time.  Section \ref{S:bars} describes our $O(n\log n)$-time algorithm for polygons without spurs, but possibly with forks.  We give a complete description of the algorithm of Cortese \etal\ for closed walks in planar graphs, reformulated using our terminology, in Section \ref{S:segments} and then describe and analyze our faster implementation in Section \ref{S:faster}.  Due to space constraints, many details and proofs are deferred to the appendix.  
In particular, we relate weak simplicity to related notions of compound planarity and self-touching linkage configurations in Appendix \ref{S:rigid}; we provide implementation details and proofs of correctness for an important \emph{expansion} operation in Appendix \ref{S:expansion}; and we describe several extensions and open problems in Appendix \ref{S:outro}.

\section{Background}
\label{S:back}

\subsection{Curves and Polygons}

A \EMPH{path} in the plane is a continuous function $P\colon [0,1] \to \R^2$.   
A \EMPH{closed curve} in the plane is a continuous function $P\colon S^1 \to \R^2$.  
A path or closed curve is \EMPH{simple} if it is injective.

A \EMPH{polygon} is a piecewise-linear closed curve; every polygon is defined by a cyclic sequence of points, called \EMPH{vertices}, connected by line segments, called \EMPH{edges}.  
We often describe polygons by listing their vertices in parentheses; for example, \EMPH{$(p_0, p_1, \dots, p_{n-1})$} denotes the polygon with vertices~$p_i$ and edges $p_i p_{i+1 \bmod n}$ for every index $i$.  
A polygon is simple if and only if its vertices are distinct and its edges intersect only at common endpoints.  We emphasize that the vertices of a non-simple polygon need not be distinct, and two edges of a non-simple polygon may overlap or even coincide \cite{m-ggfpi-1769,g-pmrpw-12}.  
However, we do assume without (significant) loss of generality that every edge of a polygon has positive length.

Similarly, a \EMPH{polygonal chain} is a piecewise-linear path, which consists of a linear sequence of points (vertices) connected by line segments (edges).  
We often describe polygonal paths by listing their vertices in square brackets, like \EMPH{$[p_0, p_1, \dots, p_{n-1}]$}, to distinguish them from polygons.  
A \EMPH{corner} of a polygon or polygonal chain $P$ is a subpath $[p_{i-1}, p_i, p_{i+1}]$ consisting of two consecutive edges of $P$.

Three local features of polygons play important roles in our results.
A \EMPH{spur} of a polygon $P$ is a vertex of $P$ whose two incident edges overlap.
A \EMPH{fork} of a polygon $P$ is a vertex of~$P$ that lies in the interior of an edge of $P$.
Finally, a \EMPH{simple crossing} between two paths is a point of transverse intersection; if the paths are polygonal chains, this intersection point could be a vertex of one or both paths.

\subsection{Distances}

\def\Frechet{d_{\mathcal{F}}}
\def\Vdist{d_{V}}
\def\Hdist{d_{\mathcal{H}}}

Let \EMPH{$d(p,q)$} denote the Euclidean distance between two points $p$ and $q$ in the plane.  The \EMPH{Fréchet distance $\Frechet(P,Q)$} between two closed curves $P$ and $Q$ is defined as
\[
\Frechet(P,Q) := \inf_{\rho\colon S^1\to S^1}\,  \max_{t\in S^1} \, d(P(\rho(t)) , Q(t)),
\]
where the infimum is taken over all orientation-preserving homeomorphisms of $S^1$.  
The function $\rho$ is often called a \EMPH{reparametrization} of $S^1$.  
Fréchet distance is a complete metric over the space of all (unparametrized) closed curves in the plane. 
The Fréchet distance between paths is defined similarly.

For any two polygons $P = (p_0, p_1, \dots, p_{n-1})$ and $Q = (q_0, q_1, \dots, q_{n-1})$ with the same number of vertices, the \EMPH{vertex distance} between $P$ and $Q$ is defined as
\[
	\Vdist(P,Q) := \min_s \max_i \, d(p_i , q_{i+s\bmod n}).
\]
For each integer $n$, vertex distance is a complete metric over the space of all $n$-vertex polygons (where two polygons that differ only by a cyclic shift of indices are identified).

\subsection{Planar Graphs}

A \EMPH{planar embedding} of a graph represents the vertices by distinct points in the plane, and the edges by interior-disjoint simple paths between their endpoints.  
A \EMPH{planar straight-line graph} is a planar graph embedding in which every edge is represented by a single line segment.  
We refer to the vertices of a planar straight-line graph as \EMPH{nodes} and the edges as \EMPH{segments}, to distinguish them from the vertices and edges of polygons.  
Any planar graph embedding partitions the plane into several regions, call the \EMPH{faces} of the embedding. 
Euler's formula states that any planar embedding of a connected graph with $V$ vertices and $E$ edges has exactly $2-V+E$ faces.

Any planar graph embedding can be represented abstractly by a \EMPH{rotation system}, which records for each node the cyclic sequence of incident segments \cite{mt-gs-01}.  Planar straight-line graphs can be represented by several data structures, each of which allows fast access to the abstract graph, the coordinates of its nodes, and the rotation system; one popular example is the doubly-connected edge list \cite{bcko-cgaa-08}.

\subsection{Weak Simplicity}
\label{S:weak-simple}
Intuitively, a closed curve or polygon is \emph{weakly simple} if it can be made simple by an arbitrarily small perturbation, or equivalently, if it is the limit of a sequence of simple closed curves or polygons.  
Our two different metrics for curve similarity give us two different formal definitions:
\begin{itemize}
\item
A closed curve $P$ is \EMPH{weakly simple} if, for any $\e > 0$, there is a simple closed curve $Q$ such that $\Frechet(P,Q) < \e$.  
In other words, a closed curve is weakly simple if it can be made simple by an arbitrarily small perturbation of the entire curve.
\item
A polygon~$P$ is \EMPH{rigidly weakly simple} if, for any $\e > 0$, there is a simple polygon $Q$ with the same number of vertices such that $\Vdist(P,Q) < \e$.  
In other words, a polygon is rigidly weakly simple if it can be made simple by an arbitrarily small perturbation of its vertices.
\end{itemize}

For any two polygons $P$ and $Q$ with the same number of vertices, we have $\Frechet(P,Q) ≤ \Vdist(P,Q)$; thus, every rigidly weakly simple polygon is also weakly simple.  The following theorem, whose proof we defer to Appendix \ref{S:rigid}, implies that the two definitions are \emph{almost} equivalent for polygons. 
\begin{theorem}
\label{Th:rigid}
Every weakly simple polygon with more than two vertices is rigidly weakly simple.
\end{theorem}

Our proof relies on a nontrivial result of Ribó Mor~\cite[Theorem~3.1]{r-rcpps-06} (previously conjectured by Connelly \etal~\cite[Conjecture~4.1]{cdr-ilstl-02}) about “self-touching” linkage configurations.  The restriction to polygons with more than two vertices is necessary; every polygon with at most two vertices is weakly simple, because it is a degenerate ellipse, but not rigidly weakly simple, because every simple polygon has at least three vertices.

\subsection{Earlier Definitions of Weak Simplicity}
\label{S:earlier}

\begin{wrapfigure}{r}{3.125in}
\vspace{-6ex}
\includegraphics[scale=0.4]{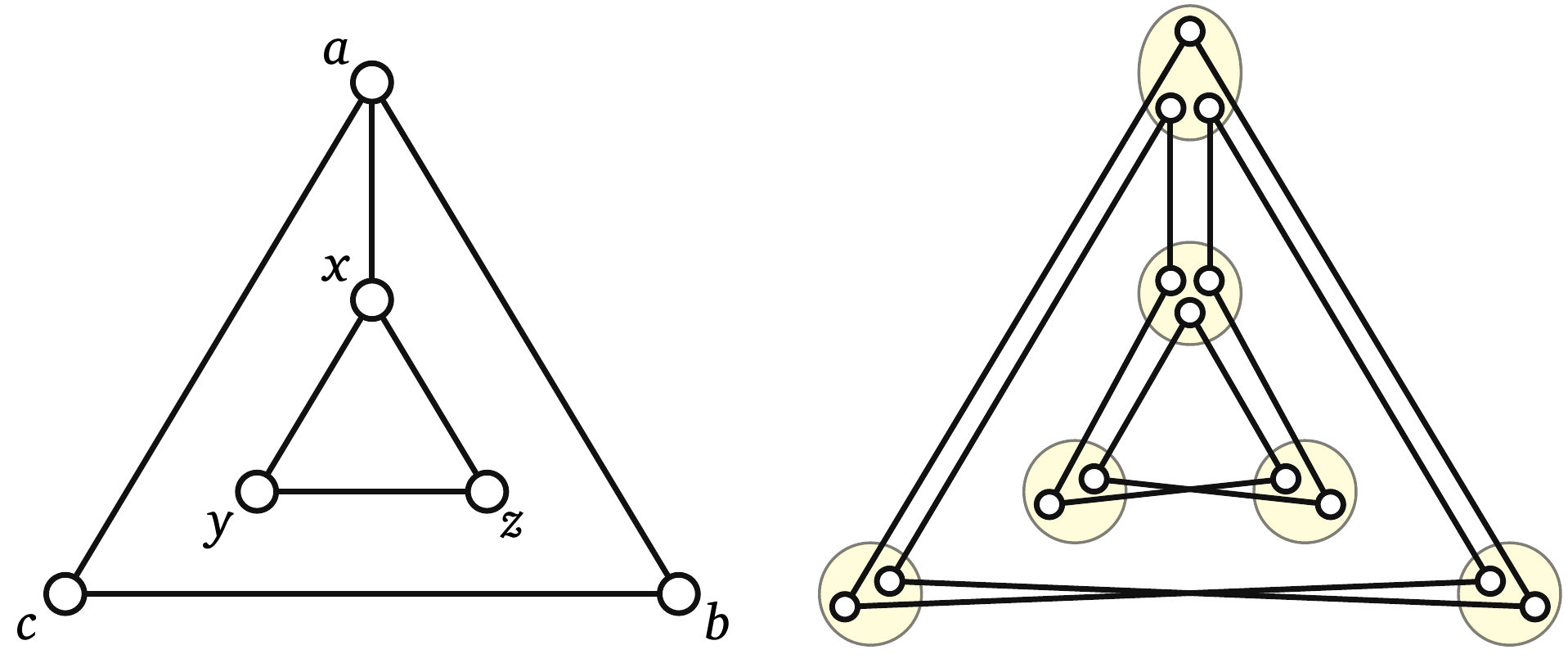}
\centering
\caption{A polygon that many published definitions incorrectly classify as weakly simple.}
\label{F:just-bad-polygon}
\end{wrapfigure}

Several authors have offered combinatorial definitions of weakly simple polygons based on the intuitive notation that a weakly simple polygon cannot cross itself; however, \emph{every} such definition we have found in the literature is either imprecise, incomplete, or incorrect.  For example, a common definition of “weakly simple”, originally due to Toussaint \cite{t-cgpis-89, t-sspst-89, st-csp-92, b-pvgrp-93}, requires that the rotation number (the sum of signed external angles divided by~$2\pi$) is either $+1$ or $-1$; however, rotation numbers are not well-defined for polygons with spurs.  \emph{Every} combinatorial definition of “(self-)crossing” we have found in the literature \cite{st-csp-92, ylw-spnprp-97, noncrossing, t-cgpis-89, ksn-spppo-99, do-gfalo-07, b-pvgrp-93} is incorrect for polygons with spurs; many of these definitions are  incorrect even for polygons without spurs.  Figure~\ref{F:just-bad-polygon} shows a spur-free polygon with 14 vertices that is not weakly simple but satisfies several published definitions of weak simplicity, including Toussaint's.  We offer further details, examples, and discussion in Appendix \ref{S:earlier2}.

We emphasize that the \emph{algorithms} in all these papers appear to be correct.  Given simple polygons as input, these algorithms construct weakly simple polygons as intermediate results, with additional structure that is consistent with the papers' definitions.  More importantly, in each case, the perturbations required to make those polygons simple are implicit in their construction.  Despite occasional claims to the contrary (based on overly restrictive definitions) \cite{aghtu-acgg-08}, we are unaware of any previous algorithm to determine whether a polygon is weakly simple.

\section{No Spurs or Forks}
\label{S:nodes}

In this section, we describe an algorithm to determine whether a polygon without spurs or forks is weakly simple in $O(n\log n)$ time.
Although we have not seen a complete description of this algorithm in the literature, it is essentially folklore.  Similar techniques were previously used by Reinhart \cite{r-ajccs-62}, Zieschang \cite{z-aekf-65}, and Chillingworth \cite{c-wns2-72} to determine whether a given closed curve in an arbitrary 2-manifold with boundary is homotopic to a simple closed curve, and more recently by Cortese \etal\ \cite{cbpp-ecpg-09} to determine clustered planarity of closed walks in plane graphs.

Let $P = (p_0, \dots, p_{n-1})$ be an arbitrary polygon without spurs or forks.  
The algorithm consists of three phases.  
First, we construct the image of $P$, to identify all coincident vertices and edges, and to rule out simple crossings between edges of $P$.  Second, we look for simple crossings at the vertices of $P$. 
Finally, if there are no simple crossings, we expand $P$ into a nearby 2-regular plane graph in the only way possible, and then check whether the expansion is consistent with~$P$.

The description and analysis of our algorithm use the following multiplicity functions: \EMPH{$n(u)$} denotes the number of times the point $u$ occurs as a vertex of $P$; \EMPH{$n(uv)$} denotes the number of times the line segment $uv$ occurs (in either orientation) as an edge of $P$; and \EMPH{$n(uvw)$} denote the number of times the corner $[u,v,w]$ or its reversal $[w,v,u]$ occurs in~$P$.

\subsection{Constructing the Image Graph}

In the first phase, we determine in $O(n\log n)$ time whether any two edges of $P$ cross, using the classical sweep-line algorithm of Shamos and Hoey~\cite{sh-gip-76}; if so, we immediately halt and report that $P$ is not weakly simple.  
Otherwise, the image of~$P$ is a planar straight-line graph $G$, whose vertices we call \emph{nodes} and whose edges we call \emph{segments}, to distinguish them from the vertices and edges of $P$.  
Specifically, the nodes of $G$ are obtained from the vertices of $P$ by removing duplicates, and the segments of $G$ are obtained from the edges of $P$ by removing duplicates.  
The image graph $G$ can be constructed in $O(n\log n)$ time using, for example, a straightforward modification of Shamos and Hoey's sweep-line algorithm.  
This is the most time-consuming portion of our algorithm; the other two phases require only $O(n)$ time.\footnote{We conjecture that the image graph of a weakly simple polygon can actually be constructed in $O(n\log^* n)$ or even $O(n)$ time, by a suitable modification of algorithms for triangulating simple polygons \cite{ctv-flvat-89,s-sfira-91,d-rysoa-92,c-tsplt-91,agr-lttsp-01}.}

\subsection{Node Expansion}

\begin{wrapfigure}{r}{3.1in}
\vspace{-5ex}
\centering
\includegraphics[scale=0.28]{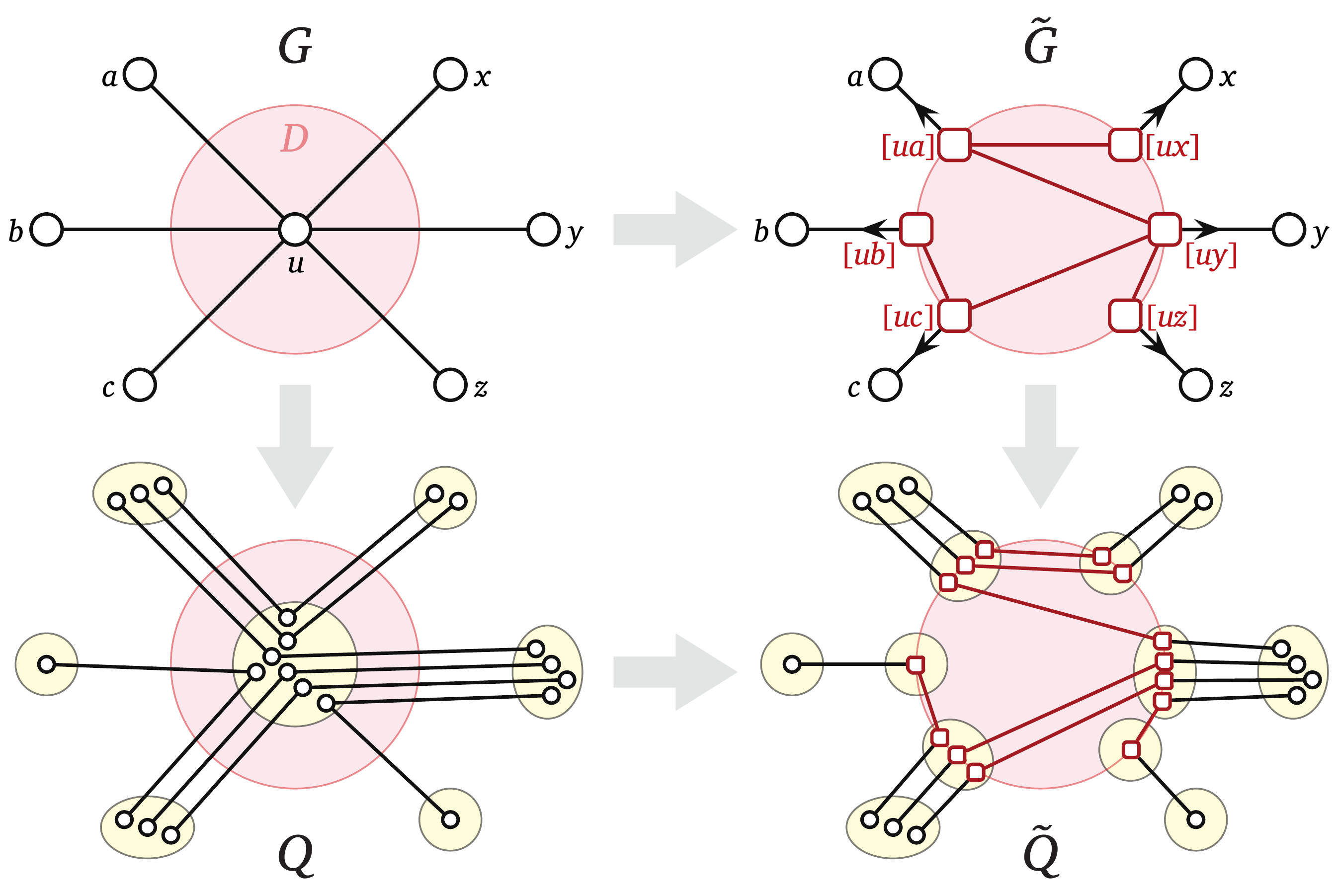}
\vspace{-1ex}
\caption{Node expansion.  Compare with Figure \ref{F:expand}.}
\label{F:node-expand}
\vspace{-1ex}
\end{wrapfigure}

In the second phase, since we have already ruled out simple crossings in the interiors of edges, any simple crossing in $P$ must consist of two corners $[u,v,w]$ and $[x,v,y]$ whose endpoints have the cyclic order $u,x,w,y$ around their common middle vertex~$v$.  
We can detect all such crossings in $O(n)$ time using an operation we call \EMPH{node expansion}, defined geometrically as follows.  (Cortese \etal\ \cite{cbpp-ecpg-09} call this operation a \emph{cluster expansion}.)

Let $u$ be an arbitrary node in $G$.  
Consider a disk~$D_u$ of radius $\delta$ centered at $u$, with $\delta$ chosen sufficiently small that $D_u$ intersects only the segments of~$G$ incident to $u$.  
Because edges of $G$ are straight line segments, the circle $\bdry D_u$ intersects each edge at most once.  
On each segment $ux$ incident to $u$, we introduce a new node $[ux]$ at the intersection point $ux \cap \bdry D_u$.  
Then we modify $P$ by replacing each maximal subwalk in $D_u$ with a straight line segment.  
Thus, if $P$ contains the subpath $[x,u,y]$, the modified polygon contain the edge $[ux][uy]$.  
Let $\tilde{P}$ denote the modified polygon, and let $\tilde{G}$ denote the image of $\tilde{P}$ in the plane; see the top row of Figure~\ref{F:node-expand}.

Lemmas \ref{L:expand} and \ref{L:node} (in the appendix) imply that the original polygon $P$ is weakly simple if and only if the expanded polygon $\tilde{P}$ is weakly simple.  If $P$ has a simple crossing at $u$, then some pair of edges of $\tilde{P}$ cross inside $D_u$.  We describe how to check for such crossings in  Section \ref{SS:expand-planar}.

Altogether, expanding node $u$ and checking for simple crossings requires $O(n(u))$ time; thus, we can expand \emph{every} node of $G$ in $O(n)$ time overall.  
If any node expansion creates a simple crossing, we halt immediately and report that $P$ is not weakly simple. 
Otherwise, let $\tilde{P}$ denote the polygon after all node expansions, and let~$\tilde{G}$ denote its image graph.  By induction, $\tilde{G}$ is a plane graph, and $P$ is weakly simple if and only if $\tilde{P}$ is weakly simple.

\subsection{Global Inflation}
\label{SS:inflation}

In the final phase of the algorithm, we “inflate” the image graph $\tilde{G}$ into a 2-regular plane graph $\tilde{Q}$ by replacing each segment $s$ of $\tilde{G}$ with several parallel segments, one for each edge of $\tilde{P}$ that traverses $s$; see the right column of Figure \ref{F:node-expand}.  Then $\tilde{P}$ (and therefore $P$) is weakly simple if and only if $\tilde{Q}$ is connected and consistent with $\tilde{P}$.

Fix an arbitrarily small positive real number $\e\ll \delta$.  To construct~$\tilde{Q}$ from~$\tilde{G}$, we replace each node $[uv]$ with $n(uv)$ closely spaced points on $\bdry D_u$, all within $\e$ of $[uv$].  Then we replace each segment $[uv][vu]$ of $\tilde{G}$ outside the disks with $n(uv)$ parallel segments.  Finally, within each disk $D_u$, we replace each corner segment $[uv][uw]$ with $n(uvw)$ parallel segments, so that all resulting segments in $D_u$ are disjoint.  Crucially, there is exactly one way to perform this final replacement within each disk $D_u$ once the boundary points are fixed; see \cite{splitting,surfcut,bs-gwbin-85} for related constructions.

The uniqueness of $\tilde{Q}$ implies $\tilde{P}$ is weakly simple if and only if (1) the resulting 2-regular graph~$\tilde{Q}$ is connected and thus a simple polygon, and (2) we have $\Vdist(\tilde{P}, \tilde{Q}) < \e$.  We can check whether $\tilde{Q}$ is connected in $O(n)$ time by depth-first search.  Finally, if $\tilde{Q}$ is a simple polygon, we can determine whether $\Vdist(\tilde{P}, \tilde{Q}) < \e$ in $O(n)$ time using any fast string-matching algorithm.

\medskip
One final detail remains to complete the proof of Theorem \ref{Th:nospur+nofork}: How do we choose appropriate parameters $\delta$ and $\e$ in the second and third phases?  In fact, there is no need to choose specific values at all! The combinatorial embedding of the expanded image graph $\tilde{G}$ is identical for all sufficiently small positive $\delta$; similarly, the combinatorial embedding of the 2-regular graph $\tilde{Q}$ is identical for all sufficiently small positive $\e \ll \delta$.  Thus, the entire algorithm can be performed by modifying abstract graphs and their rotation systems, without choosing explicit values of $\delta $ or $\e$ at all!  See Section \ref{SS:expand-planar} for details.
\begin{theorem}
\label{Th:nospur+nofork}
Given an arbitrary $n$-vertex polygon $P$ without spurs or forks, we can determine in $O(n\log n)$ time whether $P$ is weakly simple.
\end{theorem}

\section{Forks But No Spurs}
\label{S:bars}

Recall that a \EMPH{fork} in a polygon is a vertex that lies in the interior of an edge.  With only minor modifications, the algorithm in the previous section can also be applied to polygons with forks, but still without spurs.  Specifically, in the preprocessing phase, we locate all forks in $O((n+k)\log n)$ time using a standard sweep-line algorithm \cite{bo-arcgi-79}, where $k = O(n^2)$ is the number of forks, and then subdivide each edge into smaller edges at the forks on that edge.  The remainder of the algorithm is unchanged.  In the worst case, subdividing the polygon to eliminate forks increases its complexity from $n$ to $\Theta(n^2)$, which in turn increases the overall running time of the algorithm from $O(n\log n)$ to $O(n^2\log n)$.  With more care, however, we can avoid this quadratic blowup.

We define a coarser decomposition of the image of~$P$ into points and line segments called a \EMPH{bar decomposition}, as follows.  A \EMPH{bar} of $P$ is a component of the union of the interiors of all edges of $P$ that lie on a common line.  Every fork lies in exactly one bar; we call any vertex that is not in a bar \EMPH{sober}.  
\begin{wrapfigure}{r}{3in}
\centering
\includegraphics[scale=0.3]{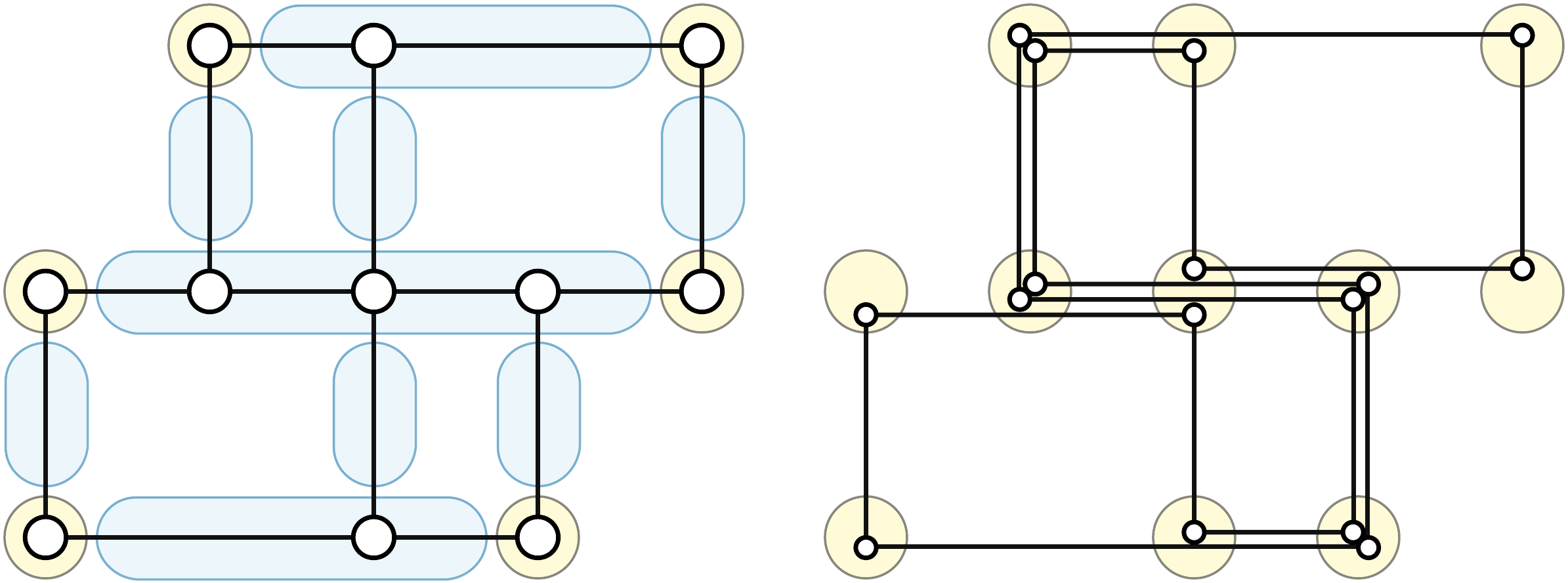}
\caption{A bar decomposition of a weakly simple polygon without spurs, and a nearby simple polygon.}
\vspace{-1ex}
\end{wrapfigure}
The bar decomposition consists of all bars of $P$ and all distinct sober vertices of $P$; the image of~$P$ is equal to the union of these bars and points.  If $P$ has no forks, the bar decomposition consists of the nodes and segments of the image graph $G$.

We can compute the bar decomposition of $P$ in $O(n\log n)$ time as follows.  First, cluster the edges of~$P$ into collinear subsets, by sorting them by the slopes and $y$-intercepts of their supporting lines.  Then for each collinear subset, sort the endpoints by $x$-coordinate, breaking ties by sorting all right endpoints before all left endpoints.  Finally, a linear-time scan along each sorted collinear subset of edges yields all bars that lie on that line.  We also identify all distinct sober vertices, compute (the indices of) all vertices and edges of $P$ that lie in each bar, and compute the bars incident to each fork in cyclic order.  Finally, we verify that no pair of bars crosses, using a standard sweep-line algorithm \cite{sh-gip-76}.

Next we perform a node expansion at each sober node, as described in the previous section.  We also perform a \EMPH{bar expansion} around each bar, defined as follows.  For any bar $b$, let $b^\circ$ denote the subset of~$b$ obtained by removing all points within distance $2\delta$ of endpoints of $b$, and let $D_b$ denote an elliptical disk whose major axis is $b^\circ$ and whose minor axis has length $2\delta$, where the parameter $\delta$ is chosen so that $D_b$ intersects only the edges of~$P$ that also intersect the bar.  Then bar expansion is identical to node expansion: we subdivide $P$ at the intersection points $\im P\cap \bdry D_b$, replace each maximal subpath of $P$ that lies in $D_b$ with a line segment, and finally verify that the new segments do not cross, as described in  Section~\ref{SS:expand-planar}.

Again, there is no need to choose a specific value of $\delta$; each bar expansion is performed entirely combinatorially.  Since every edge of $P$ touches at most three bars or sober nodes, we can expand every bar and every sober node in $O(n)$ time.  The resulting polygon $\tilde{P}$ has no simple crossings, no forks, and no spurs.  Lemmas \ref{L:expand}, \ref{L:node}, and \ref{L:bar} imply inductively that $\tilde{P}$ is weakly simple if and only if the original polygon $P$ is weakly simple.  Finally, we can determine whether~$\tilde{P}$ is weakly simple in $O(n)$ time using the global inflation algorithm described in Section \ref{SS:inflation}.  We conclude:
\begin{theorem}
\label{Th:nospurs}
Given an arbitrary $n$-vertex polygon $P$ without spurs, but possibly with forks, we can determine in $O(n\log n)$ time whether $P$ is weakly simple.
\end{theorem}

\section{Polygons with Spurs}
\label{S:segments}

Neither of the previous algorithms is correct when the input polygon has spurs.  To handle these polygons, we apply a recent algorithm of Cortese \etal~\cite{cbpp-ecpg-09} to determine whether a walk in a plane graph is weakly simple (in their terminology, whether a rigid clustered cycle is c-planar).  We include a complete description of the algorithm here, reformulated in our terminology, not only to keep this paper self-contained, but also so that we can improve its running time in Section \ref{S:faster}. 

At a high level, the algorithm proceeds as follows.  In an initial preprocessing phase, we construct the image graph $G$ of the input polygon $P$ and then expand $P$ at every node in $G$.  Then, following Cortese \etal~\cite{cbpp-ecpg-09}, we repeatedly modify the polygon using an operation we call \emph{segment expansion}, defined in Section \ref{SS:segment-expansion}, until either we detect a simple crossing (in which case $P$ is not weakly simple), the image graph is reduced to a single line segment (in which case $P$ is weakly simple), or there are no more “useful” segments to expand.

If the input polygon has no forks, the image graph $G$ has complexity $O(n)$, we construct it in $O(n\log n)$ time, and we expand every node in $O(n)$ time, exactly as described in Section \ref{S:nodes}.  A simple potential argument implies that the main loop terminates after at most $O(n)$ segment expansions.  (Cortese \etal~\cite{cbpp-ecpg-09} prove an upper bound of $O(n^2)$ expansions using a different potential function.)  A~naïve implementation executes each segment expansion in $O(n)$ time, giving us an overall running time of $O(n^2)$.  The $O(n\log n)$ time bound  follows from a more careful implementation and analysis, which we describe in Section \ref{S:faster}.
\begin{theorem}
\label{Th:spurs}
Given an arbitrary $n$-vertex polygon $P$ without forks but possibly with spurs, we can determine in $O(n\log n)$ time whether $P$ is weakly simple.
\end{theorem}

Any polygon can be subdivided into a polygon without forks in $O(n^2\log n)$ time using a standard sweep-line algorithm, similar to the algorithm described in Section \ref{S:bars}.  Thus, Theorem \ref{Th:spurs} has the following immediate corollary.
\begin{corollary}
\label{Th:spurs+forks}
Given an arbitrary $n$-vertex polygon $P$, we can determine in $O(n^2\log n)$ time whether $P$ is weakly simple.
\end{corollary}

For the remainder of the paper, we assume that the input polygon $P$ has no forks.

\subsection{Segment Expansion}
\label{SS:segment-expansion}

We now describe the main loop of our algorithm in more detail.  The main operation, \EMPH{segment expansion}, is nearly identical to our earlier expansion operations.  Let $s$ be an arbitrary segment of the current image graph $G$.  Let $D_s$ be an elliptical disk whose boundary intersects precisely the segments that share one endpoint with $s$.  To expand segment $s$, we subdivide $P$ at its intersection with $\bdry D_s$, replace each maximal subpath of $P$ that lies in $D_s$ with a line segment, and verify that the new segments do not cross.

Unfortunately, expanding an \emph{arbitrary} segment could change a polygon that is not weakly simple into a polygon that is weakly simple.  Our algorithm expands only segments with the following property.  A segment $uv$ in the image graph is a \EMPH{base} of one of its endpoints $u$ if every occurrence of $u$ in the polygon~$P$ is immediately preceded or followed by the other endpoint~$v$.  A segment is \EMPH{safe} if it is a base of both of its endpoints.  Cortese \etal\ prove that a polygon $P$ is weakly simple if and only if the polygon~$\tilde{P}$ that results from expanding a \emph{safe} segment is also weakly simple  \cite[Lemma~3]{cbpp-ecpg-09}.  We provide a new proof of this key lemma in the appendix (Section \ref{S:seg-correctness}).

After a node expansion around any node $u$, each new node $[ux]$ on the boundary of $D_u$ has a base, namely the segment $[ux]x$ \cite[Property~5]{cbpp-ecpg-09}.  Thus, expanding every node in the original image graph guarantees that every node in the image graph has a base.  Similarly, after any segment expansion, every newly created node has a base.  Thus, our algorithm inductively maintains the invariant that every node in the image graph has a base.  It follows that at every iteration of our algorithm, the image graph has at least one safe segment, so our algorithm never gets stuck.

\subsection{Useful Segments}

However, under some circumstances, expanding a safe segment does not actually make progress.  Specifically, if both endpoints of the segment have degree $2$ in $G$, and no spur in $P$ includes that segment, the segment expansion does not change the combinatorial structure of $P$ or $G$ at all.  We call any such segment \EMPH{useless}.  Equivalently, a safe edge is \EMPH{useful} if it is the only base of one of its endpoints, and useless otherwise.  Our algorithm repeatedly expands only \emph{useful} segments until either the image graph consists of a single segment (in which case $P$ is weakly simple) or there are no more useful segments to expand. 
\begin{lemma}
\label{L:end}
Let $G$ be the image graph of a polygon $P$ without simple crossings.  If every node of $G$ has a base but $G$ has no useful segments, then $G$ is a simple polygon and $P$ has no spurs.
\end{lemma}

\begin{proof}
Let $H$ be the subgraph of all safe segments of $G$.  The degree of any node in $H$ is at most the number of bases of that node in $G$.  Since no node in $G$ can have more than two bases, $H$ is the union of disjoint paths and cycles.  If some component of $H$ is a path, the first (or last) segment in that path is useful.  Otherwise, every node in $G$ has two bases, and therefore has degree $2$; because $G$ is a connected planar straight line graph, it must be a simple polygon.  Moreover, because every node in $G$ has two bases, $P$ has no spurs.
\end{proof}
Lemma \ref{L:end} implies that when our main loop ends, we can invoke our algorithm for spur-free polygons in Section \ref{S:nodes}. But this is overkill.  Because $P$ has no spurs, $P$ must traverse every segment of $G$ the same number of times.  It follows that the $2$-regular plane graph $\tilde{Q}$ constructed in the inflation phase of the algorithm in Section \ref{S:nodes} will consist of $n(uv)$ parallel copies of $G$.  Thus, when the main loop ends, $P$ is weakly simple if and only if $n(uv) = 1$, for any single segment $uv$.

\subsection{Termination and Analysis}

Like Cortese \etal, we prove that our algorithm halts using a potential argument.  Let $\abs{P}$ and $\abs{G}$ respectively denote the number of vertices in polygon $P$ and the number of nodes in its image graph~$G$; our potential function is $\Phi(P,G) := \abs{P} - \abs{G}$.   Clearly $\Phi(P,G)$ is always non-negative.  Our preprocessing phase at most doubles $\abs{P}$, so at the beginning of the main loop we have $\Phi(P, G) \le 2n$.  (Cortese \etal\ used the potential $\sum_{e} n(e)^2 = O(n^2)$, where the sum is over all segments of $G$; otherwise, our analysis is nearly identical.)
\begin{lemma}
\label{L:potential}
Let $P$ be any polygon whose image graph $G$ has more than one segment, let $\tilde{P}$ be the result of expanding a useful segment of $G$, and let $\tilde{G}$ be the image graph of $\tilde{P}$.  Then $\Phi(\tilde{P}, \tilde{G}) < \Phi(P,G)$.
\end{lemma}

\begin{proof}
Let $uv$ be a useful segment of $G$, where without loss of generality we have $\deg(u) \le \deg(v)$.  Because $uv$ is safe, every maximal subpath of $P$ that lies in the ellipse $D_{uv}$ contains at least $2$ vertices, so $\abs{\tilde{P}} \le \abs{P}$.  We easily observe that $\abs{\tilde{G}} = \abs{G} + \deg(u) + \deg(v) - 4$, so if $\deg(u) + \deg(v) > 4$, the proof is complete.  There are three other cases to consider:
\begin{itemize}
\item If $\deg(u) = \deg(v) = 1$, then $uv$ is the only segment of $G$, which is impossible.
\item Otherwise, if $\deg(u) = 1$, then $P$ must contain the spur $[v,u,v]$, which implies $\abs{\tilde{P}} \le \abs{P} - 2$ and therefore $\Phi(\tilde{P}, \tilde{G}) \le \Phi(P,G) - (\deg(v) - 1)\le \Phi(P,G) - 1$.
\item Finally, suppose $\deg(u) = \deg(v) = 2$.  Because $uv$ is useful, $P$ must contain one of the spurs $[v,u,v]$ or $[u,v,u]$, which implies $\abs{\tilde{P}} \le \abs{P} - 2$ and therefore $\Phi(\tilde{P}, \tilde{G}) \le \Phi(P,G) - 2$. \qed
\end{itemize}
\unskip
\end{proof}
\unskip
Lemma \ref{L:potential} immediately implies the main loop of the algorithm ends after at most $2n$ segment expansions.  Because the potential $\Phi$ decreases at every iteration, the polygon never has more than $2n$ vertices.  We can find a useful segment in the image graph (if one exists) and perform a segment expansion in $O(n)$ time by brute force.  Thus, a naïve implementation of our algorithm runs in $O(n^2)$ time.

\section{Fast Implementation}

\label{sec:implement}
\label{S:faster}

Finally, we describe a more careful implementation of the previous algorithm that runs in $O(n\log n)$ time.  We build the image graph $G$ and perform the initial node expansions in $O(n\log n)$ time, just as in the previous sections.  After building some necessary data structures, we repeatedly expand useful segments until we either find a local crossing or there are no more useful segments.

\subsection{Data Structures}

Our algorithm uses the following data structures. We maintain the image graph $G$ in a standard data structure for planar straight-line graphs, such as a doubly-connected edge list \cite{bcko-cgaa-08}.
The polygon $P$ is represented by a circular doubly-linked list that alternates between vertex records and edge records; each vertex in $P$ points to the corresponding node in $G$.

Call an edge of $P$ \EMPH{simple} if neither endpoint is a spur and \EMPH{complex} otherwise.  For each segment $uv$, we separately maintain a doubly-linked list $\emph{Simple}(uv)$ of all simple edges of $P$ that coincide with $uv$, and a doubly-linked list $\emph{Complex}(uv)$ of all complex edges of $P$ that coincide with $uv$. Every record in these lists has pointers both to and from the corresponding edge record in the circular list representing~$P$.  Each of these lists also maintains its size.

Each node $u$ in $G$ also maintains a pointer to its base (or bases).  Finally, we maintain a global queue of all useful segments in $G$.  All of these data structures can be constructed in $O(n)$ time after the preprocessing phase.

\subsection{Segment Expansion}

\def\Spurs{\sigma}
\def\Notspurs{\phi}

Our algorithm divides each segment expansion into the following four phases: (1) remove all spurs at the endpoints $u$ and $v$; (2) compute the nodes $[ua]$ and $[vz]$ in cyclic order around the ellipse $\bdry D_{uv}$; (3) straighten the remaining paths through $u$ and through $v$; (4) build the graph $G_D$, check for simple crossings, and update the image graph $G$.  Our analysis uses the following functions of $u$ (and similar functions of $v$), all defined just before the segment expansion begins:
\begin{itemize}\itemsep0pt
\item \EMPH{$\deg(u)$} is the number of segments incident to $u$, including $uv$.
\item \EMPH{$n(u)$} is the number of vertices of $P$ that coincide with $u$.
\item \EMPH{$\Spurs(u)$} is the number of spurs of $P$ that coincide with $u$.
\item \EMPH{$\Notspurs(u)$}${}= n(u) - \Spurs(u)$ is the number of vertices at $u$ that are not spurs.
\end{itemize}

\paragraph{Phase 1: Remove spurs.~}
We remove all spurs at $u$ and $v$ by brute force, by traversing $uv$’s list of complex edges.  Each time we remove a spur $s$, we check the edges of $P$ immediately before and after $s$, and if necessary, move them from $\emph{Simple}(s)$ to $\emph{Complex}(s)$ or vice versa.  We also update the count $n(u)$ and $n(v)$. After this phase, every maximal subpath of $P$ inside the disk $D_{uv}$ has length at most $3$.  The total time for this phase is $O(\Spurs(u) + \Spurs(v))$.  

\paragraph{Phase 2: Compute sequence of new nodes.~}
Next, we compute the intersection points $G \cap \bdry D$ in cyclic order by considering the segments incident to $u$ in cyclic order, starting just after $uv$, and then the segments incident to $v$ in cyclic order, starting just after $uv$.  These two cyclic orders are accessible in the doubly-connected edge list representing $G$.  For each intersection point $[ua]$ or $[vz]$, we initialize a new node record with a pointer to its base $[ua]a$ or $[vz]z$.  At this point, we can update the queue of useful segments.

Finally, we find a segment $ua^*$ with maximum weight $n([ua^*])$ and a segment $vz^*$ with maximum weight $n([vz^*])$, such that $a^*\ne v$ and $z^*\ne u$.  To simplify notation, let $u’$ denote the point $[ua^*] = ua^* \cap \bdry D$ and let $v’$ denote the point $[vz^*] = vz^* \cap \bdry D$.  The total time for this phase is $O(\deg(u)  + \deg(v))$.

\paragraph{Phase 3: Expansion.~}
The third phase straightens the remaining constant-length paths through $u$ and $v$ \emph{except} for subpaths $[a^*, u, v, z^*]$ or $[a^*, u, a^*]$ or $[z^*, v, z^*]$.  Specifically, for each segment $ua$ with $a\ne a^*$, we replace subpaths of $P$ containing segment $ua$ with corresponding paths through the new node $[ua]$.  Specifically: 
\begin{itemize}\itemsep0pt
\item $[a, u, v]$ becomes $[a, [ua], v]$
\item $[a, u, a]$ becomes $[a, [ua], a]$
\item $[a, u, b]$ becomes $[a, [ua], [ub], b]$
\end{itemize}
Then for each segment $vz$ with $z\ne z^*$, we similarly replace subpaths of $P$ that contain $vz$ with paths through the new node $[vz]$.  Each path replacement takes $O(1)$ time, including the time to update the relevant simple- and complex-edge lists.

At the end of these two loops, the only remaining subpaths through $u$ and $v$ have the forms $[a^*, u, v, z^*]$ or $[a^*, u, a^*]$ or $[z^*, v, z^*]$.  To update these subpaths, we intuitively \emph{move} node $u$  to $u’$ and move node $v$ to $v’$.  But in fact, “moving” these two nodes has no effect on our data structures at all; we only change their \emph{names} for purposes of analysis.

The total time for this phase at most $O(\Notspurs(u) - n(u’) + \Notspurs(v) - n(v’))$, where $n(u’)$ and $n(v’)$ denote the number of vertices at $u’$ and $v’$ \emph{after} the segment expansion.

\paragraph{\boldmath Phase 4: Check planarity and update $G$.~}
We discover all the segments in the graph $G_D$ in the previous phase.  If there are more than $2(\deg(u) + \deg(v))$ such segments, then $G_D$ cannot be planar, so we immediately halt and report that $P$ is not weakly simple.  Otherwise, we compute the rotation system of $G_D$ and check its planarity in $O(\deg(u) + \deg(v))$ time, as described in Section \ref{SS:expand-planar}.  If $G_D$ is a plane graph, we splice it into the image graph $G$, again in $O(\deg(u) + \deg(v))$ time.

\medskip
This completes our implementation of segment expansion.

\subsection{Time Analysis and Heavy-Path Decomposition}
\label{S:heavy-path}

The total running time of a single edge expansion is at most
\[
	O(\Spurs(u) + \Spurs(v))  ~+~ O(\Notspurs(u) - n(u’) + \Notspurs(v) - n(v’)) ~+~ O(\deg(u) + \deg(v)).
\]
Since each segment incident to $u$ carries at least one edge, we have $\deg(u) < \Spurs(u) + \Notspurs(u)$, so we can charge the last term to the first two.  The total time spent removing spurs is clearly $O(n)$.  We bound the total remaining time as follows.

For purposes of analysis, imagine building a \EMPH{family tree $T$} of all nodes that our algorithm ever creates.  The root of $T$ is a special node $r$, whose children are the nodes in the initial image graph (after node expansion).  The children in $T$ of any other node $u$ are the new nodes $[ua]$ created during the  segment expansion that destroys $u$.  (Likewise, the children of $v$ are the new nodes $[vz]$.)  Each node $u$ in $T$ has weight $n(u)$, which is the number of vertices located at $u$ just before the segment expansion that destroys $u$, or equivalently, just after the segment expansion that creates $u$.  To simplify analysis, we set $n(r) = n$, the initial number of vertices in $P$.  Let $C(u)$ denote the children of node $u$ in $T$, and let $u’$ denote the maximum-weight child of $u$.  Finally, let $N$ denote the set of all nodes, and let $N’$ denote the set of all maximum-weight children.

Expanding segment $uv$ moves every vertex at $u$ that is not a spur to one of the children of $u$, and then merges pairs of coincident vertices to form spurs; thus,
\[
	\Notspurs(u) ~= \sum_{x\in C(u)} (n(x) + \Spurs(x)).
\]
It follows that
\[
	\Notspurs(u) - n(u’) ~= \sum_{x\in C(u)\setminus\set{u’}} n(x)
						+ \sum_{x\in C(u)} \Spurs(x),
\]
and therefore
\[
	\sum_{u\in N} (\Notspurs(u) - n(u’)) ~=
		\sum_{x\in N\setminus N’} n(x) + O(n).
\]
The following standard heavy-path decomposition argument\cite{ht-fafnc-84, st-dsdt-83}  implies that $\sum_{x\in N\setminus N’} n(x) = O(n\log n)$.   If we remove the vertices in $N'$ from $T$ by contracting each node in $N’$ to its parent, the resulting tree has height $O(\log n)$, and the total weight of the nodes at each level is $O(n)$.
We conclude that the total time spent expanding segments is $O(n\log n)$, which completes the proof of Theorem~\ref{Th:spurs}.



\newpage
\bibliographystyle{newuser}
\bibliography{wspolygon,jeffebib/jeffe,jeffebib/topology,jeffebib/optimization,jeffebib/compgeom}


\newpage
\appendix
\part*{Appendix}

\section{Problems with Previous Definitions}
\label{S:earlier2}

\subsection{Crossing and Self-Crossing}

Many authors have offered the intuition that a polygon is weakly simple if and only if it is not “self-crossing”; indeed, this is the complete definition offered by some authors \cite{ht-pcbet-05,aghtu-acgg-08}.  However, a proper definition of “self-crossing” is quite subtle; we believe the most natural and general definition is the following.  Two paths $P$ and $Q$ are \EMPH{weakly disjoint} if, for all sufficiently small $\e>0$, there are disjoint paths $\tilde{P}$ and $\tilde{Q}$ such that $\Frechet(P, \tilde{P}) < \e$ and $\Frechet(Q, \tilde{Q}) < \e$.  Two paths \EMPH{cross} if they are not weakly disjoint.  Finally, a closed curve is \EMPH{self-crossing} if it contains two crossing subpaths.

However, this is \emph{not} the definition of “crossing” that most often appears in the computational geometry literature; variants on the following combinatorial definition are much more common \cite{t-cgpis-89,ksn-spppo-99,do-gfalo-07,b-pvgrp-93}.  Two polygonal chains $P = [p_0, p_1, \dots, p_\ell]$ and $Q = [q_0, q_1, \dots, q_\ell]$ with length $\ell \ge 3$ have a \EMPH{forward crossing} if they satisfy two conditions:
\begin{itemize}\itemsep0pt
\item
$p_i = q_i$ for all $1\le i\le \ell-1$, and
\item
the cyclic order of $p_0, q_0, p_2$ around~$p_1$ is equal to the cyclic order of $p_\ell, q_\ell, p_{\ell-2}$ around~$p_{\ell-1}$.
\end{itemize}
Polygonal chains $P$ and $Q$ have a \EMPH{backward crossing} if $P$ and the reversal of $Q$ have a forward crossing.  (See Figure \ref{F:crossing}.)  Finally, two polygonal chains $P$ and $Q$ \EMPH{cross} if some subpaths of $P$ and $Q$ have a simple crossing, a forward crossing, or a backward crossing.

\begin{figure}[htb]
\centering
\includegraphics[scale=0.4]{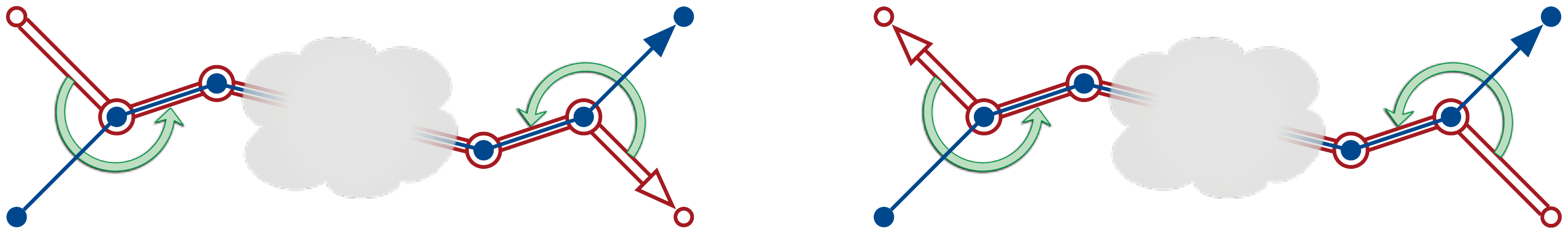}
\caption{A forward crossing and a backward crossing.  The paths coincide within the cloud.}
\label{F:crossing}
\end{figure}

These two definitions appear to be equivalent for polygonal chains \emph{without spurs}, but they are not equivalent in general; Figure \ref{F:spur-crossing} shows two simple examples where the definitions differ.   In fact, we know of no combinatorial definition of “crossing” that agrees with our topological definition for arbitrary polygonal chains with spurs.

\begin{figure}[htb]
\centering
\includegraphics[scale=0.4]{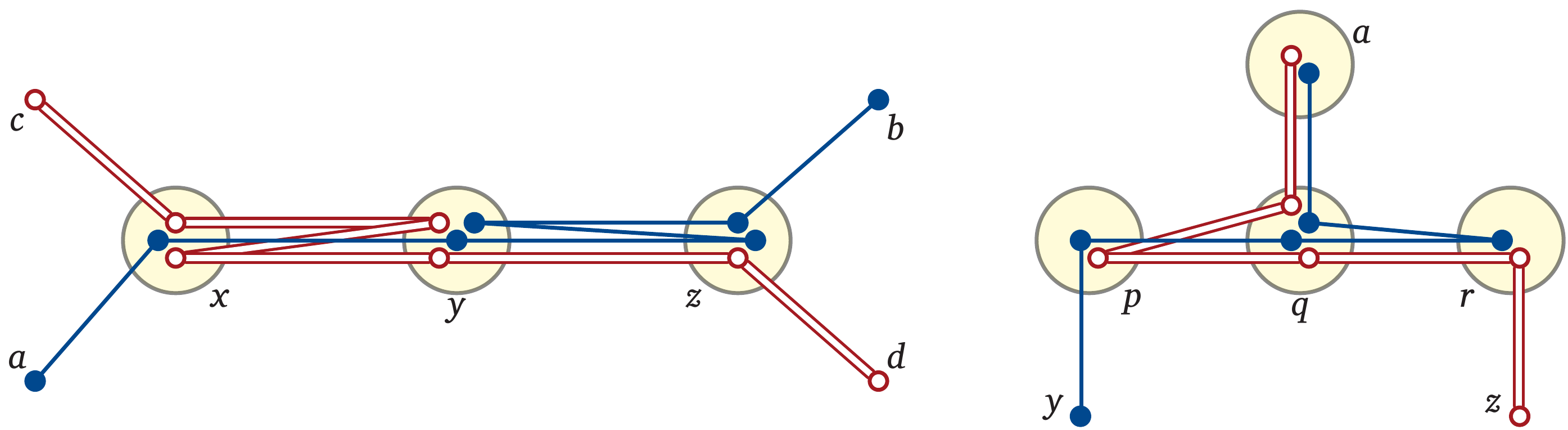}
\caption{Two pairs of crossing paths that do not satisfy the combinatorial definition.  Vertices in each small circle coincide.}
\label{F:spur-crossing}
\end{figure}

Several other authors (including the second author of this paper) have proposed the following more intuitive definition \cite{st-csp-92, ylw-spnprp-97, noncrossing}: Two paths $P$ and~$Q$ “cross” if, for arbitrarily small neighborhoods $U$ of some common subpath, path $P$ contains a point in more than one component of $\bdry U \setminus Q$; see Figure \ref{F:bad-crossing}(a).  
This definition is correct when $P$ and $Q$ are \emph{simple}, but it yields both false positives and false negatives for non-simple paths, even without spurs.  
For example, the paths $[w,x,y,z,x,y]$ and $[z,x,y,z,x,w]$ in Figure \ref{F:bad-crossing}(b) cross but do not satisfy this definition (because~$P$ has \emph{no} points in $\bdry U \setminus Q$), and the paths $[a,x,y,z,b]$ and $[c,x,y,z,d]$ in Figure \ref{F:bad-crossing}(b) do not cross but satisfy this definition (because small neighborhoods of the common subpath are not simply connected).

\begin{figure}[htb]
\centering\footnotesize\sffamily
\begin{tabular}{ccc}
	\includegraphics[scale=0.4]{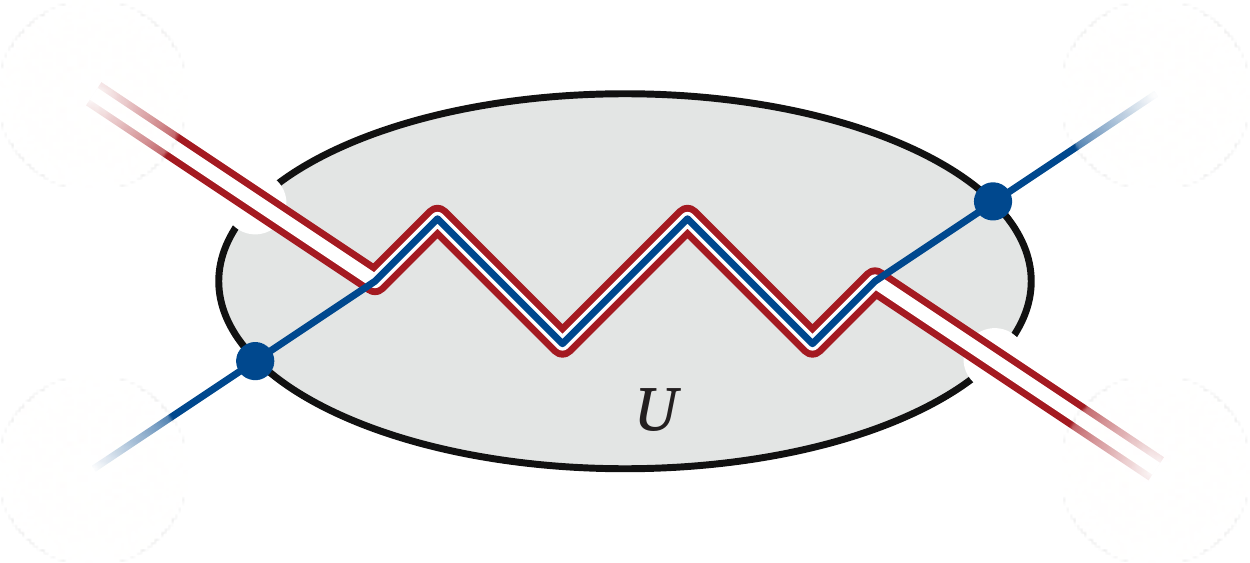} &
	\includegraphics[scale=0.4]{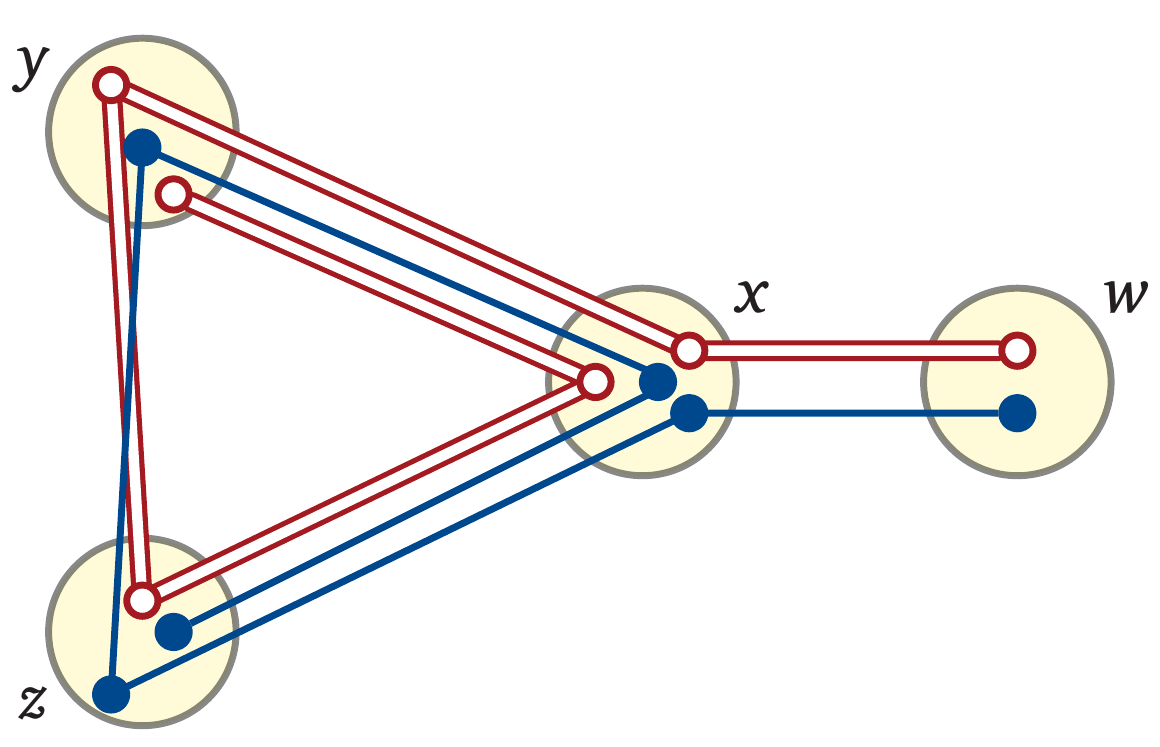} &
	\includegraphics[scale=0.4]{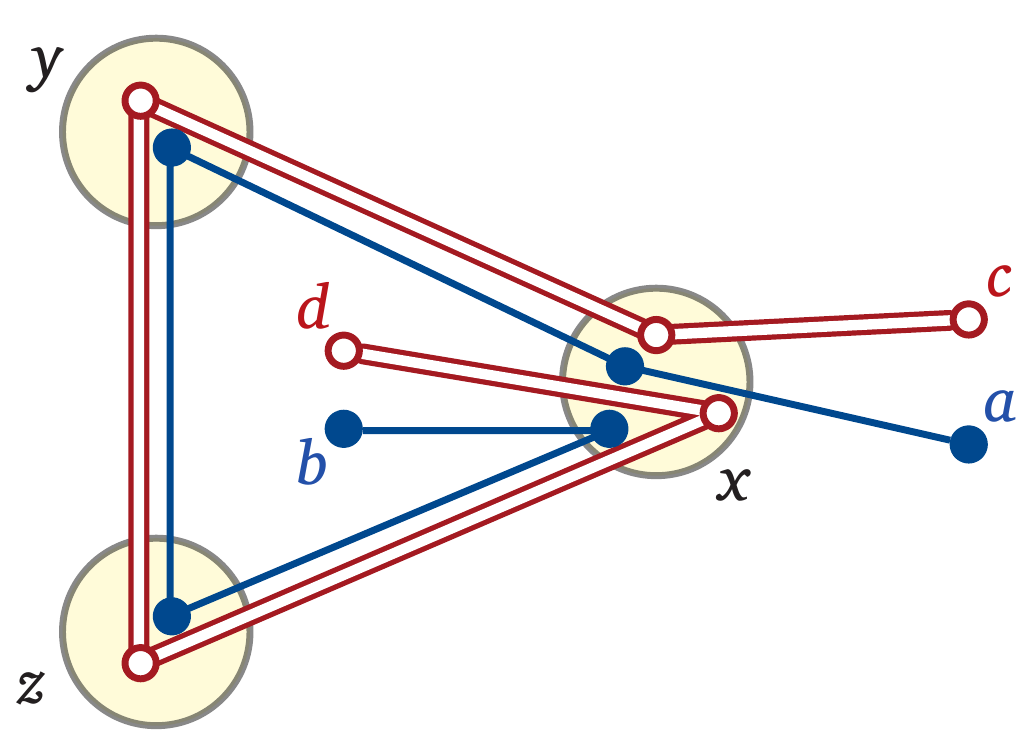}
	\\[-1ex]
	(a) & (b) & (c)
\end{tabular}
\caption{(a) A common intuitive definition of “crossing”.  (b) Crossing paths that do not satisfy this definition. (c)~Non-crossing paths that do satisfy this definition.  Vertices in each small circle coincide.}
\label{F:bad-crossing}
\end{figure}

\subsection{Weak Simplicity}

\begin{wrapfigure}{r}{1.75in}
\centering
\vspace{-2ex}
\includegraphics[scale=0.125]{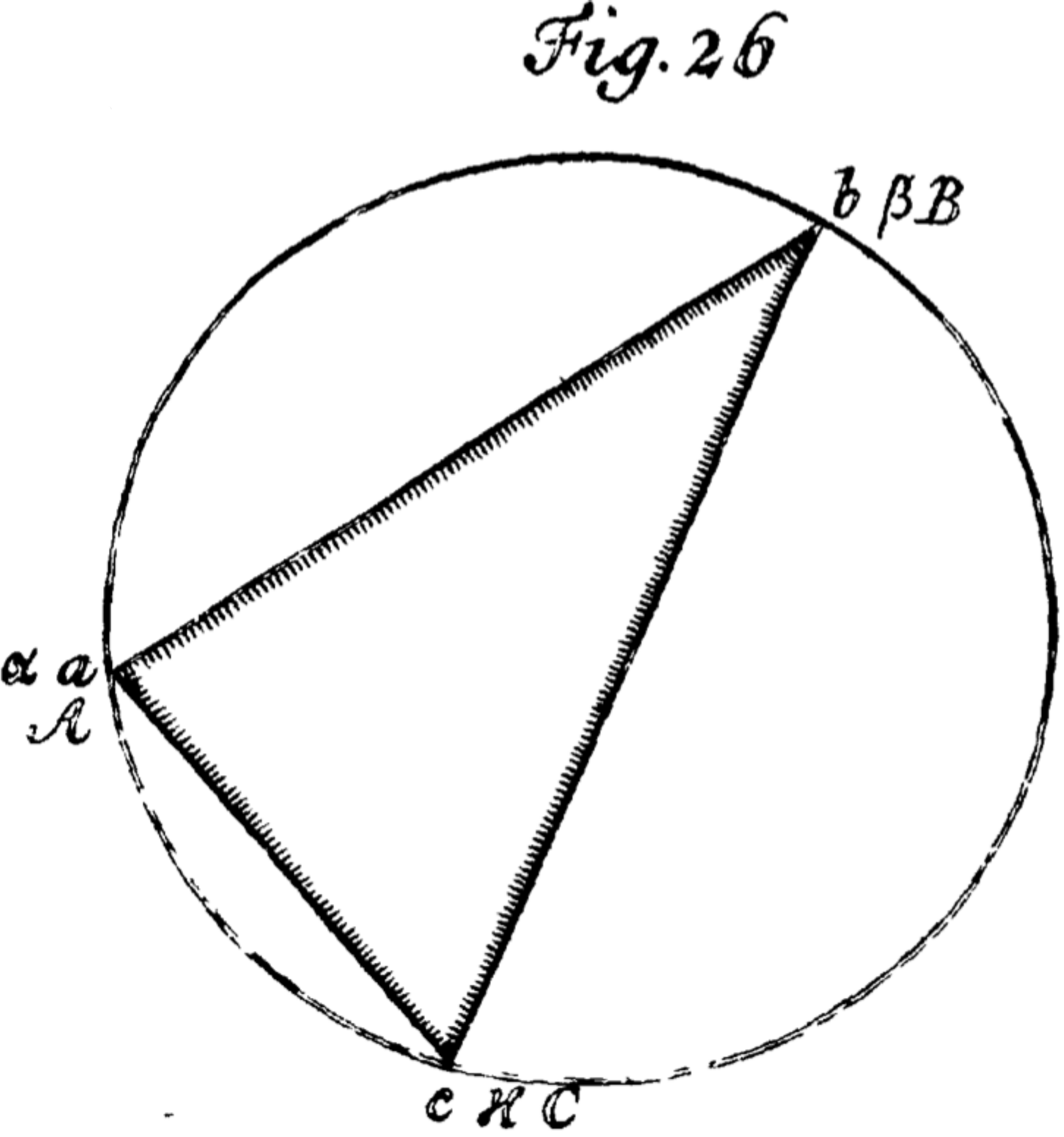}
\caption{Meister's enneagon $[a, b, c, \alpha, \beta, \kappa, A, B, C]$ is neither weakly simple nor self-crossing.}
\vspace{-1ex}
\end{wrapfigure}

Avoiding self-crossings is a \emph{necessary} condition for a polygon to be weakly simple, but it is not sufficient.  For example, a polygon that wraps 3 times around triangle (as considered by Meister \cite{m-ggfpi-1769}) and the polygon $(a,x,b,x,c,x,a,x,b,x,c,x)$ in Figure \ref{F:first-example} are neither self-crossing nor weakly simple.

Toussaint \cite{t-cgpis-89, t-sspst-89, st-csp-92, b-pvgrp-93} defines a polygon to be weakly simple if its \emph{rotation number} is either $+1$ or~$-1$ and any pair of points splits the polygon into two paths that cross.  
Here, the \EMPH{rotation number} of a polygon is the sum of the signed external angles at the vertices of the polygon, divided by $2\pi$, where each external angle is normalized to the interval $(-\pi, \pi)$.  
Unfortunately, the rotation number of a polygon with spurs is undefined, since there is no way to determine locally whether the external angle at a spur should be $\pi$ or $-\pi$.  
As a result, Toussaint's definition can only be applied to polygons without spurs.

Even for polygons without spurs, there is a more subtle problem with Toussaint's definition.  Consider the 14-vertex polygon $({a,b,c,a,b,c,a,x,\!y,z,x,\!y,z,x})$ shown in Figures~\ref{F:just-bad-polygon} and~\ref{F:bad-polygon}.  This polygon contains exactly two crossings: one between subpaths $[x,a,b,c,a,b]$ and $[c,a,b,c,a,x]$, and the other between subpaths $[a,x,\!y,z,x,\!y]$  and $[z,x,\!y,z,x,a]$.  However, both pairs of crossing subpaths overlap not just geometrically, but combinatorially, as substrings of the polygon's vertex sequence.  In short, a polygon (or polygonal chain) can cross itself without being divisible into two paths that cross each other!  Demaine and O'Rourke's definition of a self-crossing linkage configuration \cite{do-gfalo-07} has the same problem. (Their definition also does not consider the possibility of backward crossings.)

\begin{figure}[htb]
\includegraphics[scale=0.4]{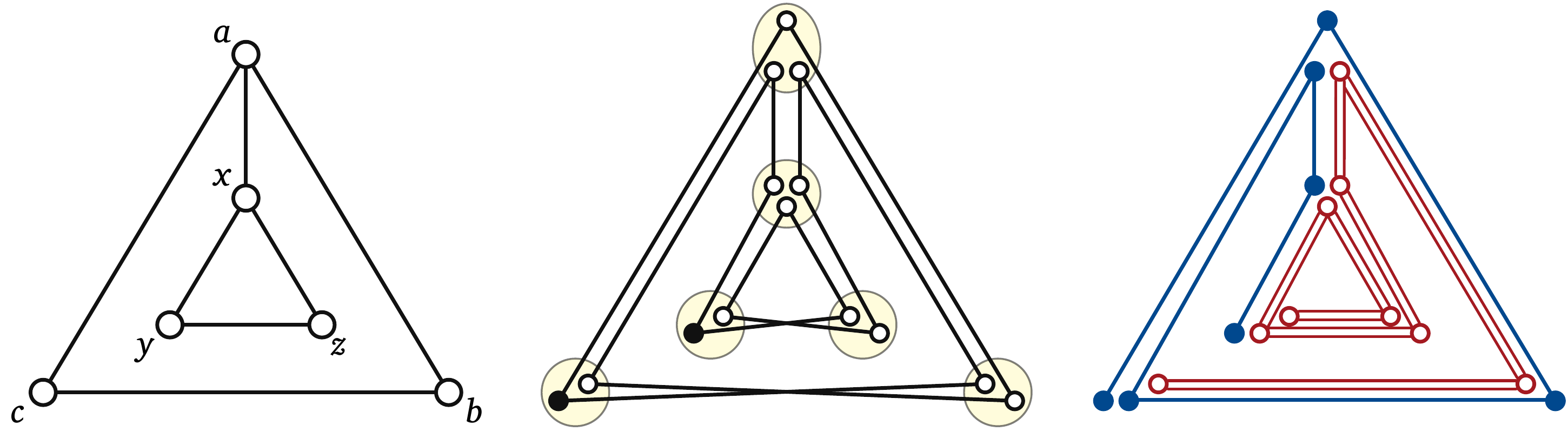}
\centering
\caption{A polygon $(a,b,c,a,b,c,a,x,y,z,x,y,z,x)$ that is not weakly simple, even though its rotation number is $1$ and every pair of vertices splits the polygon into two paths that do not cross.}
\label{F:bad-polygon}
\end{figure}

Finally, it is unclear how we could use the combinatorial definitions of “crossing” and “self-crossing” that correctly handle these subtleties to quickly determine whether a polygon is weakly simple.  
Given a polygon $P$ without spurs, there is a natural algorithm to determine whether $P$ is self-crossing in $O(n^3)$ time---Find all maximal coincident subpaths via dynamic programming, and check whether each one is a forward or backward crossing---but this is considerably slower than the algorithms we derive from the  topological definition of “weakly simple”.

Several other papers offer definitions of “weakly simple” that are overly restrictive \cite{rw-acpgg-12, hst-pbet-04, bm-corgg-04, mn-lpcgc-99, bcdt-cordg-14}. For example, the LEDA software library \cite{mn-lpcgc-99} defines a polygon to be weakly simple it it has coincident vertices but otherwise disjoint edges;  Bose and Morin \cite{bm-corgg-04, bcdt-cordg-14} define a polygon to be weakly simple if the graph~$G$ defined by its vertices and edges is plane, the outer face of~$G$ is a cycle, and one bounded face of $G$ is adjacent to all vertices.  It is straightforward to construct weakly simple polygons that contradict these definitions.

Finally, several recent papers define weak simplicity directly in terms of vertex perturbation, what we are calling \emph{rigid} weak simplicity \cite{dt-loncg-09, hst-pbets-10, ist-dcgm-13, dfor-cpst-10}, without mentioning any connection to the more general (and arguably more natural) definition in terms of Fréchet distance.\footnote{As of June 2014, Wikipedia \cite{w-sp-14} offers two definitions of “weakly simple polygon”, both of which are completely wrong.  In particular, every \emph{planar straight-line graph} is the limit, in the Hausdorff metric, of a sequence of simple polygons with a fixed number of vertices.}

Again, we emphasize that the \emph{algorithms} in all the papers we discuss in this section appear to be correct, and the simple polygons they construct are consistent with the corresponding papers' definitions.

\section{Proof of Theorem \ref{Th:rigid}}
\label{S:rigid}

In this section, we prove Theorem \ref{Th:rigid}: Every weakly simple polygon with more than two vertices is rigidly weakly simple.
In fact we will prove a stronger result; as explained in Sections~\ref{S:earlier} and \ref{S:earlier2}, most of the definitions of ``weakly simple'' in the past are either incorrect or restricted.  
However, several related notions turn out to be equivalent to our formal definition of weak simplicity.
Here is the statement of the full theorem we are about to prove; we defer the definitions of the bold terms  to later subsections.
\begin{theorem}
\label{Th:eq}
Let $P$ be a polygon with more than two vertices.  
The following statements are equivalent:
\begin{enumerate}[(a)]\cramped
\item $P$ is weakly simple.
\item $P$ is a \textbf{compound-planar rigid clustered cycle} \cite{cbpp-ecpg-09}.
\item $P$ is \textbf{strip-weakly simple}.
\item $P$ has a \textbf{self-touching configuration} \cite{cdr-ilstl-02}.
\item $P$ is rigidly weakly simple.
\end{enumerate}
\end{theorem}

These terms are roughly ordered from least restrictive to most restrictive.  The backward implications (e)$\Rightarrow$(d)$\Rightarrow$(c)$\Rightarrow$(b) and (c)$\Rightarrow$(a) all follow immediately from the definitions.  Most of the forward implications also follow directly from the definitions and the classical Jordan-Schönflies theorem: Every simple closed curve is isotopic to a circle.  The only exceptions are (a)$\Rightarrow$(c), which requires a more careful topological argument, and (d)$\Rightarrow$(e), which relies on a nontrivial result of Ribó Mor~\cite[Theorem~3.1]{r-rcpps-06}.  


\subsection{Strip system}

Let $G$ be the graph formed by the image of polygon $P$ in the plane, whose vertices we call \emph{nodes} and whose edges we call \emph{segments}.  For any real number $\e>0$, the \EMPH{$\e$-strip system} of $P$ is a decomposition of a neighborhood of $G$ into the following disks and strips.
\begin{itemize}
\item
For each node $u$ of $G$, let $D_u$ denote the disk of radius $\e$ centered at $u$.  
\item
For each segment $uv$ of $G$, let $S_{uv}$ denote the \EMPH{strip} of points with distance at most $\e^2$ from $uv$ that do not lie in the interior of $D_u$ or $D_v$.
\end{itemize}
The circular arcs $A_{u, v} = S_{uv}\cap D_u$ and $A_{v, u} = S_{uv}\cap D_v$ are called  the \EMPH{ends} of $S_{uv}$.  We assume $\e$ is sufficiently small that these disks and strips are pairwise disjoint except that each strip intersects exactly two disks at its ends.
Finally, let $U_\e$ denote the union of all these disks and strips.

We say that a polygon $P$ is \EMPH{strip-weakly simple} if, for every sufficiently small $\e > 0$, there is a simple closed curve $Q’$ inside the neighborhood $U_\e$ that crosses the disks and strips of the strip system in the same order that $P$ traverses the nodes and segments of $G$.  Formally, if $P = (p_0,p_1,\dots,p_{n-1})$, then the curve $Q’$ intersects ends only transversely, in the cyclic order
\[
A_{p_0,p_1},~ A_{p_1,p_0},~ A_{p_1,p_2},~ A_{p_2,p_1},~ \dots,~ A_{p_{n-1},p_0},~ A_{p_0,p_{n-1}}.
\]
In particular, $Q’$ never intersects the same end consecutively more than once.
Informally, we say that such a curve \EMPH{respects} the strip system of $P$.

Any closed curve $Q$ that respects the $\e$-strip system of $P$ satisfies the inequality $\Frechet(P, Q) < \e$.  It follows immediately that if $P$ is strip-weakly simple, then $P$ is also weakly simple.  The converse implication requires a more careful topological argument.
\begin{lemma}
\label{L:eq-ws-ss}
A polygon $P$ is weakly simple if and only if $P$ strip-weakly simple.
\end{lemma}

\begin{proof}
Let $P = (p_0, p_1, \dots, p_{n-1})$ be a weakly simple polygon.  By breaking the edges if necessary, we assume that $P$ has no forks.  
Fix a sufficiently small real number $\e > 0$, and let $Q$ be a simple closed curve such that $\Frechet(P,Q) < \e^2$.  
This closed curve lies within the neighborhood $U_\e$ but does not necessarily have the correct crossing pattern with the arcs of the strip system.  
To complete the proof, we show how to locally modify $Q$ into a simple closed curve $Q’$ that respects the strip system of $P$.

We call a subpath of $Q$ \EMPH{good} if it lies entirely within a strip and has one endpoint on each end of that strip.  If $\e$ is sufficiently small, $Q$ has exactly $n$ good subpaths.  We call the endpoints of the good subpaths \EMPH{good points}.  Removing the good subpaths of $Q$ leaves exactly $n$ \EMPH{bad} subpaths; each bad subpath intersects only one disk $D_u$, but may intersect any end on $\bdry D_u$ an arbitrary (or even uncountably many!) number of times.

For each node $u$, let $\widehat{D}_u$ denote the complement of the unbounded component of the complement the union of $D_u$ with all bad subpaths with endpoints in $D_u$.  The subspace $\widehat{D}_u$ is a closed topological disk, and for any two nodes $u$ and $v$, the disks $\widehat{D}_u$ and $\widehat{D}_v$ are disjoint.  The Jordan-Schönflies theorem implies that there is a homeomorphism $h_u \colon \widehat{D}_u \to D_u$ for each node $u$; without loss of generality, this homeomorphism fixes every good point on $\bdry D_u$.  

Let $D_u^-$ denote the disk of radius $\e-\e^2$ centered at $u$, and let $h^-_u \colon \widehat{D}_u \to D^-_u$ be the homeomorphism obtained by applying $h_u$ and then scaling around $u$.  
Then, we compose the  homeomorphism $h^-_u$ into a single homeomorphism $h^- \colon \bigsqcup_u \widehat{D}_u \to \bigsqcup_u D^-_u$.

\begin{figure}[htb]
\centering
\includegraphics[scale=0.5]{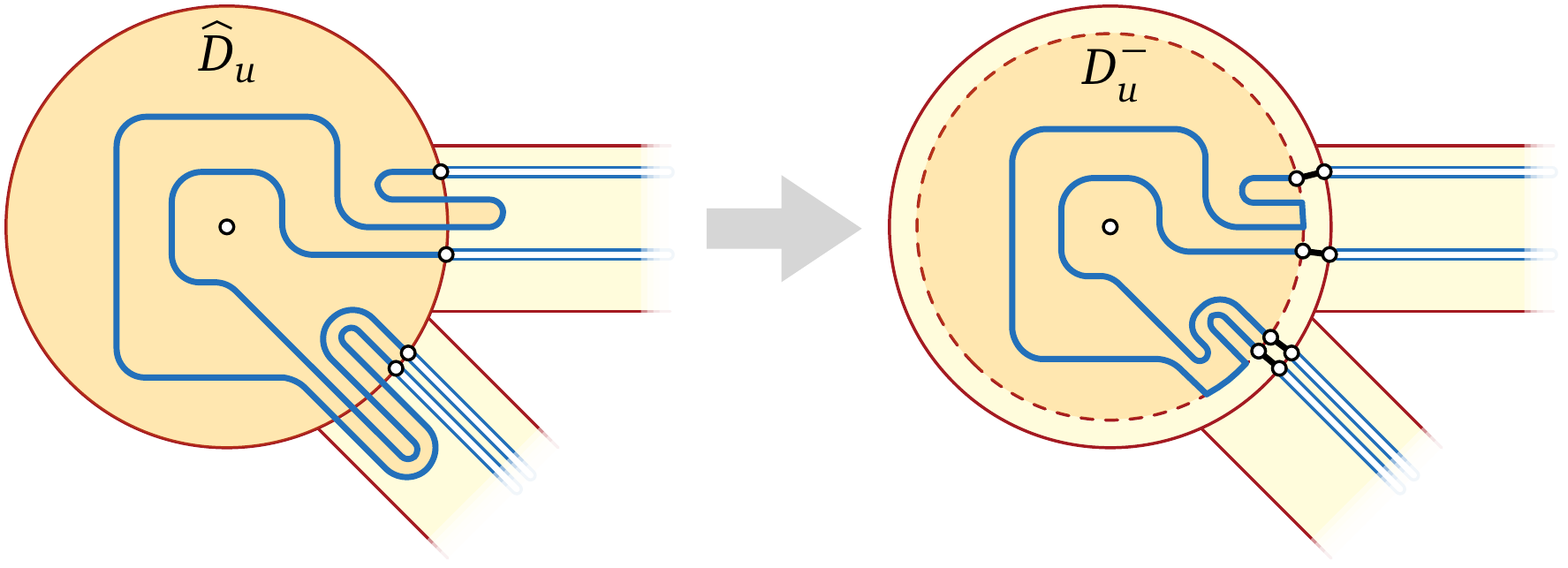}
\\[-1ex]
\caption{Subpaths of $Q$ and $Q'$ within the strip system of node $u$.}
\label{F:strip-ws}
\end{figure}

Finally, let $Q’$ be the closed curve obtained from $Q$ by replacing each bad subpath with its image under the homeomorphism $h^-$ and then, for each node $u$, connecting each good point on $\bdry D_u$ with the corresponding point on $\bdry D^-_u$ with a line segment; see Figure \ref{F:strip-ws}.  $Q’$ is a simple closed curve that respects the strip system of $P$.  It follows that $P$ is strip-weakly simple.  
\end{proof}

\subsection{Compound-Planarity and Self-Touching Configurations}

It remains to define the terms in statements (b) and (d) in Theorem \ref{Th:eq}.  Both of these terms were previously defined using different and somewhat more cumbersome language.  However, both definitions turn out to be \emph{almost} equivalent to our definition of strip-weakly simple.  We describe here only the relevant differences; we refer the reader to the original papers \cite{cbpp-ecpg-09, cdr-ilstl-02} for the original definitions.

Following Cortese \etal~\cite{cbpp-ecpg-09, fce-pcg-95, f-rcpg-96}, a polygon $P$ can be represented as a \EMPH{compound-planar rigid clustered cycle} if it respects an arbitrary \emph{topological} strip system.  In a topological strip system, the regions $D_u$ and strips $S_{uv}$ are arbitrary closed topological disks that contain the corresponding nodes and segments of $G$, where for any segment $uv$, the intersection $S_{uv} \cap D_u$ is a simple path, and otherwise all the disks are disjoint.  The Jordan-Schönflies theorem implies that there is a homeomorphism of the plane to itself that maps any topological strip system of $P$ to the $\e$-strip system of $P$, for any $\e>0$.   This gives us the equivalent (b)$\Leftrightarrow$(c).

Following Connelly, Demaine and Rote \cite{cdr-ilstl-02}, a polygon $P$ can be represented as a \EMPH{self-touching configuration} if, for any $\e > 0$, there is a simple closed curve $Q$ that respects the $\e$-strip system of $P$, with the additional requirement that for each segment $uv$, the intersection $Q \cap S_{uv}$ is a set of disjoint line segments from one end of $S_{uv}$ to the other.  
Given any closed curve $Q’$ that respects the $\e$-strip system of $P$, the Jordan-Schönflies theorem implies there there is a homeomorphism $h_{uv}\colon S_{uv} \to S_{uv}$ that straightens the good subpaths $Q\cap S_{uv}$.  
This gives us the implication  (c)$\Leftrightarrow$(d).

\subsection{Rigidification Lemma}
Finally, a self-touching configuration has a \emph{$\delta$-perturbation} if there is a planar configuration of the same linkage such that corresponding joints (vertices of the linkage) in the two configurations have distance at most $\delta$ \cite{cdr-ilstl-02}.  Equivalently, a $\delta$-perturbation of $P$ is a simple polygon at vertex distance at most $\delta$ from $P$.  Ribó Mor \cite[Theorem~3.1]{r-rcpps-06} proved that every self-touching configuration of a linkage with at least three vertices has a $\delta$-perturbation, for any $\delta>0$.  This theorem provides the final link (d)$\Rightarrow$(e), completing the proof of Theorem \ref{Th:eq}.

\section{Expansion}
\label{S:expansion}

The node expansion, bar expansion, and segment expansion operations used by the algorithms in Sections~\ref{S:nodes}, \ref{S:bars}, and \ref{S:segments} are all special cases of a more general operation, defined as follows.  Let $P = (p_0, p_1, \dots, p_{n-1})$ be an arbitrary polygon, and let $D$ be an elliptical disk that intersects $P$ \emph{transversely}; that is, the boundary ellipse $\bdry D$ intersects at least one edge of $P$, is not tangent to any edge of $P$, and does not contain any vertex of $P$. To \EMPH{expand $P$ inside $D$}, we subdivide the edges of $P$ that intersect $\bdry D$ by introducing new vertices at the intersection points, and then replace each maximal subpath of $P$ inside~$D$ with a straight line segment between its endpoints on $\bdry D$.    In the rest of this section, we provide missing proofs and implementation details for the expansion operations in our earlier algorithms.

\subsection{Preserving Weak Simplicity}

\begin{lemma}
\label{L:expand}
Let $P$ be a weakly simple polygon, and let $D$ be an elliptical disk whose boundary intersects~$P$ transversely.  The polygon $\tilde{P}$ obtained by expanding $P$ inside $D$ is weakly simple.
\end{lemma}

\begin{proof}
Suppose $P = (p_0, p_1, \dots, p_{n-1})$ is a weakly simple polygon.  By Theorem~\ref{Th:rigid}, for any real number $\e>0$, there is a simple polygon $Q =  (q_0, q_1, \dots, q_{n-1})$ such that $\Vdist(P,Q) < \e$.  If $\e$ is sufficiently small, $\bdry D$ also intersects $Q$ transversely; moreover, $\bdry D$ intersects an edge of $Q$ if and only if it intersects the corresponding edge of $P$ in the same number of points.   Let $\tilde{Q}$ be the polygon obtained by expanding $Q$ inside $D$.  See Figure \ref{F:expand}.

\begin{figure}[htb]
\centering
\includegraphics[scale=0.3]{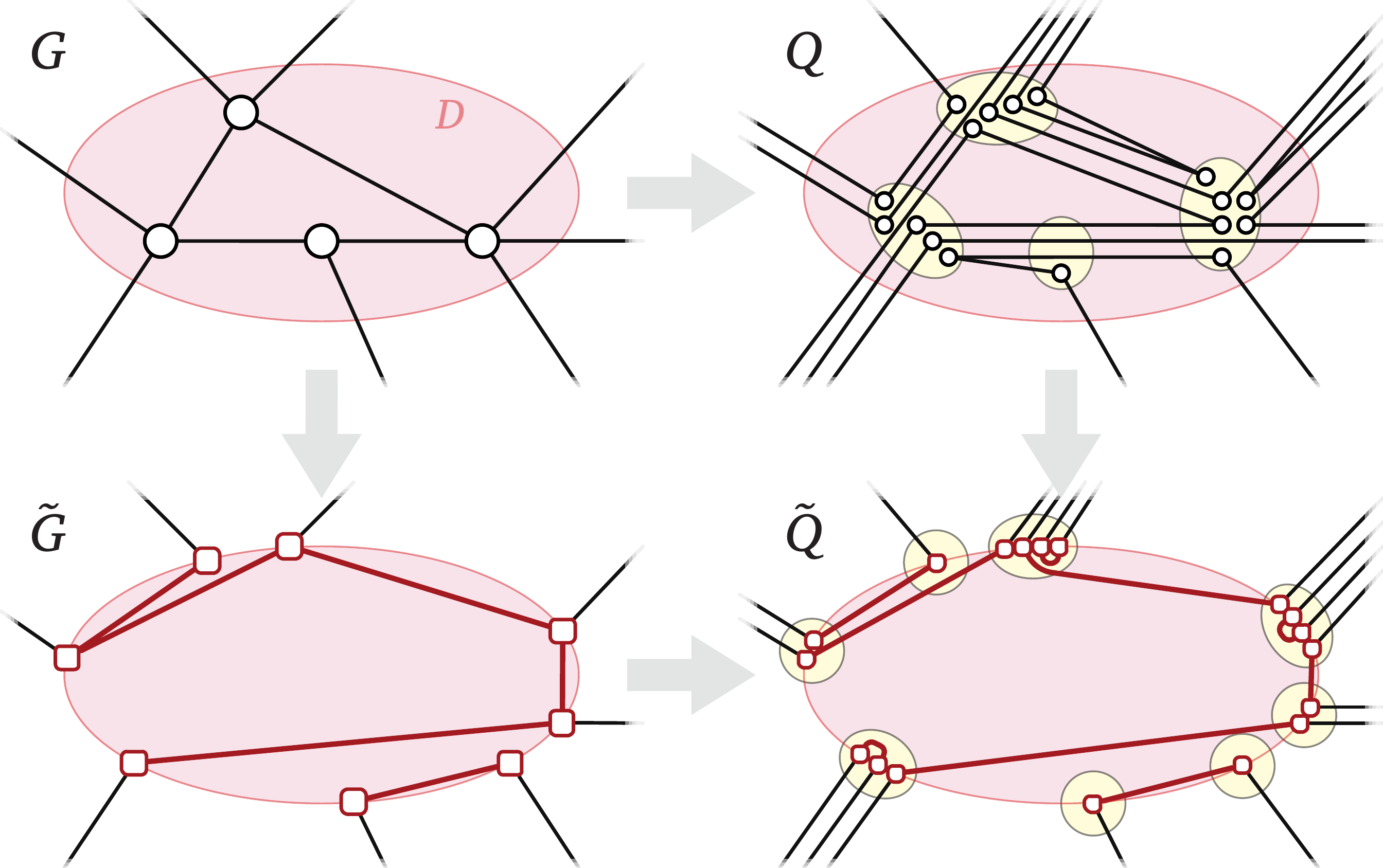}
\caption{Expansion in an ellipse. Left column: Expanding the image graph $G$.  Top row: If $P$ is weakly simple, there is a nearby simple polygon $Q$.  Right column: Expanding $Q$ in the same ellipse yields a  simple polygon $\tilde{Q}$. Bottom row: $\tilde{Q}$ is close to~$\tilde{P}$.  (Some edges of $\tilde{Q}$ are curved in the figure to make them visible.)  Compare with Figure \ref{F:node-expand}.}
\label{F:expand}
\end{figure}

Because $Q$ is simple, the intersection $Q\cap D$ consists of disjoint simple subpaths.  For any two such subpaths $\alpha$ and $\beta$, the endpoints of $\alpha$ must lie in one component of $\bdry D \setminus\beta$, and thus the corresponding line segments $\tilde\alpha$ and $\tilde\beta$ of $\tilde{Q}$ are also disjoint. It follows that $\tilde{Q}$ is also simple.

Consider a vertex $\tilde{p}$ of $\tilde{P}$, located at the intersection of an edge $p_ip_{i+1}$ of $P$ and the ellipse $\bdry D$.  Let~$\theta$ denote the angle between $p_ip_{i+1}$ and (the line tangent to) $\bdry D$ at $\tilde{p}$.  If $\e$ is sufficiently small, the corresponding edge $q_iq_{i+1}$ of $Q$ intersects $\bdry D$ at a point $\tilde{q}$ such that $d(\tilde{p},\tilde{q}) < 2\e/\sin\theta$.  It follows that $\Vdist(\tilde{P}, \tilde{Q}) < 2\e/\sin\theta^*$, where $\theta^*$ is the minimum angle of intersection between $P$ and $\bdry D$.

Thus, for any $\delta>0$, we obtain a simple polygon $\tilde{Q}$ such that $\Frechet(\tilde{P}, \tilde{Q}) < \delta$ by setting $\e < (\delta/2)\sin\theta^*$.  We conclude that $\tilde{P}$ is weakly simple.
\end{proof}
The converse of this lemma is not true in general; Figure \ref{F:bad-expansion} shows a simple counterexample.  However, as we argue below, the converse of this lemma is true for the specific expansions performed by our algorithms.

\begin{figure}[htb]
\centering
\includegraphics[scale=0.25]{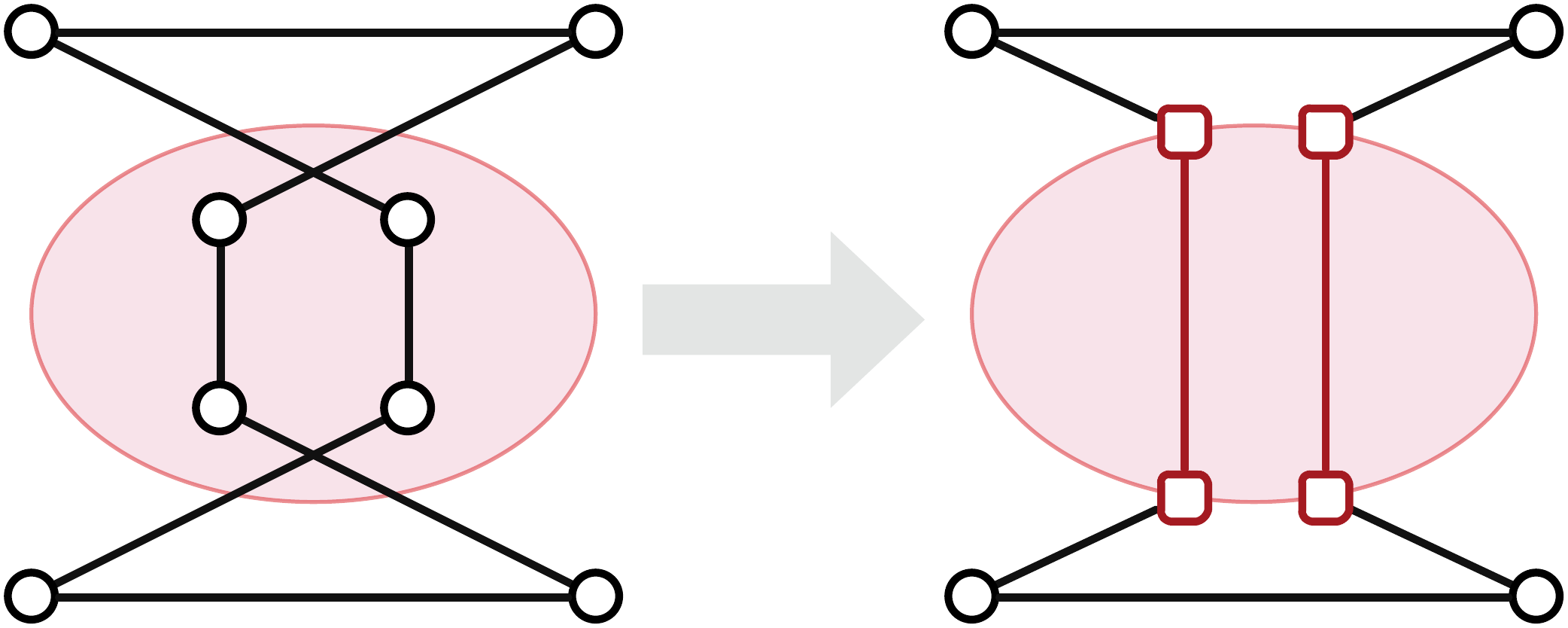}
\caption{Careless expansion can make non-weakly simple polygons simple.}
\label{F:bad-expansion}
\end{figure}

\subsection{Implementation and Planarity Checking}
\label{SS:expand-planar}

Given the polygon $P$ and ellipse $D$, it is straightforward to compute the polygon $\tilde{P}$ resulting from expanding $P$ inside $D$ in $O(n)$ time by brute force. For \emph{arbitrary} expansions, computing and sorting the coordinates of the intersection points with the ellipse $\bdry D$ requires some care.  In our algorithms, however, all expansion operations can be performed combinatorially, \emph{with no numerical computation whatsoever}.  The actual size and shape of the disk $D$ is completely immaterial to our algorithms; our geometric description in terms of ellipses is intended to provide intuition and simplify our proofs.

Our algorithms maintain a representation of the polygon $P$ that allows us to compute the sequence of points where $P$ crosses $\bdry D$, in cyclic order around $D$, in constant time per intersection point.  Specifically, for the node and segment expansions in Sections \ref{S:nodes} and \ref{S:segments}, this sequence of points can be extracted from the rotation system of the image graph $G$.  For the bar expansions in Section \ref{S:bars}, this sequence can be extracted from the order of forks along each bar, and cyclic order of bars ending at each fork; both of these orders are computed as part of the bar decomposition.  Given this sequence of points, which become the new vertices of~$\tilde{P}$, we can perform the rest of the expansion in $O(m)$ time, where $m$ is the number of edges of $P$ that intersect~$D$.

If $P$ is \emph{not} weakly simple, the expanded polygon $\tilde{P}$ may include pairs of edges that cross transversely.  Let $G_D$ denote the graph whose vertices are the intersection points $\im P\cap\bdry D$ and whose edges are distinct edges of $\tilde{P}$ inside $D$ and arcs of $\bdry D$ between vertices in cyclic order.  $\tilde{P}$ contains crossing edges if and only if $G_D$ is not a plane graph.

We can determine the planarity of $G_D$ as follows.  First, add a new “apex” vertex $a$ and connect it by edges to each of the vertices of $G_D$.  The resulting abstract graph $G_D^+$ is 3-connected, and therefore has at most one planar embedding.  Moreover, $G_D^+$ is planar if and only if there is a planar embedding of $G_D$ with all vertices on a single face (which we take to be the outer face) in the correct cyclic order.  Thus, $G_D$ is a \emph{plane} graph if and only if $G_D^+$ is a \emph{planar} graph; there are several linear-time algorithms to determine whether a graph is planar \cite{ht-ept-74,bw-cespe-04,fo-ttp-12,sh-pt-99}.
Crudely, $G_D$ has at most $O(m)$ vertices and edges, so the planarity check takes at most $O(m)$ time.

\subsection{Node Expansion}

In Section \ref{S:nodes}, we define \emph{node expansion} as expansion inside a circular disk~$D_u$ of radius $\delta$ centered at some vertex $u$ of $P$, where the radius $\delta$ is sufficiently small that $D_u$ intersects only edges incident to $u$.
\begin{lemma}[Cortese \etal~{\cite[Lemma~2]{cbpp-ecpg-09}}]
\label{L:node}
Let $P$ be an arbitrary polygon, and let $\tilde{P}$ be the result of a single node expansion.  $P$ is weakly simple if and only if $\tilde{P}$ is weakly simple.
\end{lemma}

\begin{proof}
Suppose $\tilde{P}$ is weakly simple.  Fix a sufficiently small positive real number $\e < \delta$.  By Theorem~\ref{Th:rigid}, there is a simple polygon $\tilde{Q}$ such that $\Frechet(\tilde{P},\tilde{Q}) \le \Vdist(\tilde{P},\tilde{Q}) < \e$.  We easily observe that $\Frechet(\mathit{P}, \tilde{P}) < \delta$.  Thus, the triangle inequality implies $\Frechet(\mathit{P}, \tilde{Q}) < \Frechet(\mathit{P}, \tilde{P}) + \Frechet(\tilde{P},\tilde{Q}) < \delta+\e < 2\delta$.  Because $\delta$ is arbitrarily small, we conclude that $P$ is weakly simple.  Finally, Lemma \ref{L:expand} completes the proof.
\end{proof}

\subsection{Bar Expansion}

Recall from Section \ref{S:bars} that a \EMPH{bar} of $P$ is a component of the union of all edges of $P$ that lie on some line.  
Bar expansion is defined as expansion inside an ellipse $D_b$ defined by a bar $b$ as follows.  
Fix a sufficiently small real number $\delta>0$.  
Let $b^\circ$ denote the subset of $b$ containing points at distance at least $2\delta$ from the endpoints of $b$.  Finally, $D_b$ is the ellipse whose major axis is $b^\circ$ and whose minor axis has length $2\delta$.  We emphasize that the endpoints of the bar $b$ are \emph{outside} $D_b$.
\begin{figure}[htb]
\centering
\includegraphics[scale=0.3]{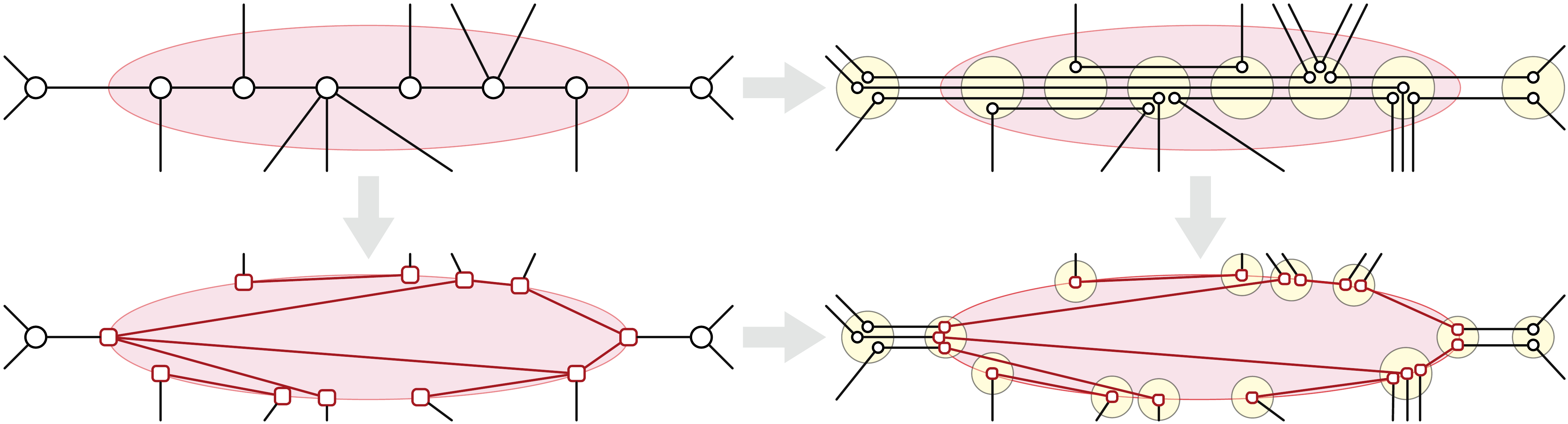}
\caption{Bar expansion.  Compare with Figure \ref{F:node-expand}.}
\label{F:bar-expand}
\end{figure}
\begin{lemma}
\label{L:bar}
Let $P$ be a polygon \textbf{without spurs}, and let $\tilde{P}$ be the result of a single bar expansion.  $P$ is weakly simple if and only if $\tilde{P}$ is weakly simple.
\end{lemma}

\begin{proof}
Let $b$ be the bar around which we are expanding, and let $\theta>0$ denote the minimum positive angle between $b$ and any edge incident to but not contained in $b$.  Consider a subpath $[u,v,w,x]$ of $P$ such that $vw$ lies in the interior of $D_b$ and therefore in the bar $b$; the vertices $u$ and $x$ must lie outside~$D_b$.  Let $\tilde{v} = uv\cap \bdry D_b$ and $\tilde{w} = wx \cap \bdry D_b$; the line segment $\tilde{v}\tilde{w}$ is an edge of the expanded polygon $\tilde{P}$.  We easily observe that $d(v, \tilde{v}) \le \delta/\sin\angle uvw \le \delta/\sin\theta$ and  $d(w, \tilde{w}) \le \delta/\sin\angle vwx \le \delta/\sin\theta$, which implies that $\Frechet([u,v,w,x], [u,\tilde{v},\tilde{w},x]) \le \delta/\sin\theta$.  By similar arguments for edges that share one or both endpoints of $b$, we conclude that $\Frechet(\mathit{P}, \tilde{P}) < \delta/\sin\theta$.

If $\tilde{P}$ is weakly simple, there is a simple polygon $\tilde{Q}$ such that $\Frechet(\tilde{P},\tilde{Q}) \le \Vdist(\tilde{P},\tilde{Q}) < \delta$, which implies $\Frechet(\mathit{P}, \tilde{Q}) < \Frechet(\mathit{P}, \tilde{P}) + \Frechet(\tilde{P},\tilde{Q}) < \delta(1+1/\sin\theta)$ by the triangle inequality.  Since this Fréchet distance can be made arbitrarily small by shrinking $\delta$, we conclude that $P$ is weakly simple.
Finally, Lemma \ref{L:expand} completes the proof.
\end{proof}
If $P$ has spurs, then bar expansions are no longer safe; the resulting polygon $\tilde{P}$ could be weakly simple even though $P$ is not.

\subsection{Segment Expansion}
\label{S:seg-correctness}

Recall from Section \ref{S:segments} that \EMPH{segment expansion} means expansion around an ellipse that just contains one segment of the image graph $G$.  
For purposes of proving correctness, we specify this ellipse more carefully as follows.  
Fix a sufficiently small real number $\delta>0$.  
For any segment $uv$ of $G$, let $uv^+$ denote the subset of all points on the line through $uv$ that have distance at most $\delta$ from $uv$.
Let $D_{uv}$ denote the ellipse whose major axis is $uv^+$ and whose minor axis has length $2\delta$.  We emphasize that in contrast to bar expansion, the endpoints of segment $uv$ lie \emph{inside} $D_{uv}$.

To simplify the proof, we imagine that each segment expansion is preceded by a \EMPH{spur reduction}, which replaces any subpath $[u,v,\dots,u,v]$ that alternates between $u$ and $v$ at least twice with the single edge $[u,v]$.  For example, the subpath $[a,u,v,u,v,u,v,b]$ would become a single spur $[a,u,v,u,b]$, and the subpath $[a,u,v,u,v,u,v,u,v,z]$ would become a simple subpath $[a,u,v,z]$.
\begin{lemma}
\label{lemma:spur-reduction}
Let $P$ be any polygon, and let $\bar{P}$ be the result of a spur reduction on some segment $uv$ of its image graph.  If $\bar{P}$ is weakly simple, then $P$ is weakly simple.
\end{lemma}

\begin{proof}
It suffices to consider the case where $\bar{P}$ is obtained from $P$ by replacing exactly one subpath $[a,u,v,u,v,z]$ with the simpler subpath $[a,u,v,z]$, for some nodes $a\ne v$ and $z\ne u$.  The argument for paths of odd length is similar, and the general case then follows by induction.

Suppose $\bar{P}$ is weakly simple.  
Then by Theorem \ref{Th:rigid} for any $\e>0$, there is a polygon $\bar{Q}$ with $\Vdist(\bar{P},\bar{Q})<\e$.  Let $[a',u',v',z']$ be the subpath of $\bar{Q}$ corresponding to the replacement subpath $[a,u,v,z]$ of $\bar{P}$.  If~$\e$ is sufficiently small, we can find four points $u_1, u_2, v_1, v_2$ with the following properties:
\begin{itemize}\cramped
\item
$d(u_1,u') = d(u_2,u') = d(v_1,v') = d(v_2,v') = \e$,
\item
$d(u_1,u'v') = d(u_2, u'v') = d(u_1,u'v') = d(u_2, u'v') = \e^2$, and 
\item
the path $[u',u_1,v_1,u_2,v_2,v']$ is simple and intersects $\bar{Q}$ only at the edge $u'v'$.
\end{itemize}
Let $Q$ be the simple polygon obtained by replacing the edge $[u',v']$ with $[u',u_1,v_1,u_2,v_2,v']$.  Then we have $\Frechet(P, Q) < 2\e$ by the triangle inequality, which implies that $P$ is weakly simple.
\end{proof}
The following lemma was proved by Cortese \etal~\cite[Lemma~3]{cbpp-ecpg-09}; we provide an alternative geometric proof here.
\begin{lemma}
\label{lemma:edge-expansion}
Let $P$ be a spur-reduced polygon with more than two distinct vertices, and let $\tilde{P}$ be the result of a safe segment expansion.  $P$ is weakly simple if and only if $\tilde{P}$ is weakly simple.
\end{lemma}

\begin{proof}
Let $uv$ be the safe segment around which we are expanding, and let $\theta$ be the smallest positive angle between $uv$ and any other segment incident to $u$ or $v$.  By the previous lemma, we can assume without loss of generality that $P$ does not contain the subpath $[u,v,u,v]$.

Suppose $\tilde{P}$ is weakly simple.  
Fix a real number $\e \ll \delta/2n$, and let $\tilde{Q}$ be a simple polygon such that $\Vdist(\tilde{P},\tilde{Q})<\e$, guaranteed by Theorem \ref{Th:rigid}.  
If $P$ has no spurs at $uv$, the proof of Lemma \ref{L:bar} implies that $\Frechet(P, \tilde{Q}) < \delta(1+1/\sin\theta)$ and we are done.  However, if $P$ has a spur at $uv$, the Fréchet distance between $P$ and $\tilde{Q}$ is approximately the length of $uv$.  In this case, we iteratively modify $\tilde{Q}$ into a new simple polygon $Q$ such that $\Frechet(P,Q) < \delta/\sin\theta + \e < \delta(1+1/\sin\theta)$.

Suppose $P$ has $k$ distinct spurs at $uv$.  For each integer $i$ from $0$ to $k$, let $D_i$ be the elliptical disk concentric with $D_{uv}$ but whose axes are shorter by $2(i+1)\e$.  For example, the major axis of $D_0$ has length $\abs{uv}+2\delta-2\e$ and the minor axis has length $2\delta-2\e$.  Every vertex of $\tilde{Q}$ lies outside each disk $D_i$, and if $\e$ is sufficiently small, every edge of $\tilde{Q}$ that intersects $D$ also intersects each disk $D_i$.

We iteratively define a sequence of polygons $\tilde{Q} = Q_0, Q_1, \dots, Q_k$ as follows.  Fix an index $i\ge 0$.  Let~$U_i$ and $V_i$ denote the subsets of $\bdry D_i$ within distance $\delta/\sin\theta$ of $u$ and $v$, respectively.  If $\delta$ is sufficiently small, the elliptical arcs~$U_i$ and $V_i$ are disjoint.  Every segment in $\tilde{Q}_i \cap D_i$ has endpoints in $U_i \cup V_i$.  We call a segment in $\tilde{Q}_i \cap D_i$ a \EMPH{left segment} if both endpoints are in~$U_i$, or a \EMPH{right segment} if both endpoints are in $V_i$.  Every left segment corresponds to a subpath of $P$ of the form $[a,u,v,u,b]$, and every right segment corresponds to a subpath of $P$ of the form $[y,v,u,v,z]$.

Suppose $\tilde{Q}_i \cap D_i$ includes at least one left segment; right segments are handled symmetrically.  Then some component $R_i$ of $D_i\setminus \tilde{Q}_i$ has both a left segment and an arc of $V_i$ on its boundary.  Suppose the left segment corresponds to the subpath $[a,u,v,u,b]$ of $P$.  Then $\tilde{P}$ contains the subpath $[a, [ua], [ub], b]$, where $[ua] = ua\cap\bdry D_{uv}$ and $[ub] = ub\cap\bdry D_{uv}$, polygon $\tilde{Q}$ contains an edge $a'b'$ such that $d(a', [ua]) < \e$ and $d(b', [ub]) < \e$, and finally $a_ib_i = a'b'\cap D_i$ is the left segment in question.  Fix a point $z_i \in R\cap V_i$, and let $\tilde{Q}_{i+1}$ be the simple polygon obtained by replacing the line segment $a_ib_i$ with the path $[a_i, z_i, b_i]$.

Arguments in the proof of Lemma \ref{L:bar} imply that $d(u, a') < \delta/\sin\theta+\e$ and $d(u, b') < \delta/\sin\theta+\e$, which in turn imply that $d(u, a_i) < \delta/\sin\theta+\e$ and $d(u, b_i) < \delta/\sin\theta+\e$.  It now follows that
\[
	\Frechet\big([[ua],u,v,u,[ub]], ~[a’, a_i, z_i, a_i, b’]\big) < \delta/\sin\theta+\e.
\]

The final polygon $Q_k$ has no left or right segments in $D_k$; thus, every spur in $P$ is within Fréchet distance $\delta/\sin\theta+\e$ of the corresponding path $Q_k$.  
We conclude that $\Frechet(P, Q_k) < \delta/\sin\theta+\e$, and Lemma~\ref{L:expand} completes the proof.
\end{proof}

\section{Generalizations and Open Problems}
\label{S:outro}

\subsection{Polygonal Chains}
\label{S:polygonal-chains}

A straightforward generalization of our algorithm can determine whether a given \emph{polygonal chain} is weakly simple in $O(n^2\log n)$ time, by checking a polygon that traverses the chain twice.
\begin{lemma}
The polygonal chain $P = [p_0, p_1, \ldots, p_{n-1}, p_n]$ is weakly simple if and only if the polygon $\hat{P} = (p_0, p_1, \dots, p_{n-1}, p_n, p_{n-1}, \dots, p_1)$ is weakly simple. 
\end{lemma}

\begin{proof}
Suppose $P$ is weakly simple.  Then for any $\delta>0$, there is a simple curve $Q$ with $\Frechet(P, Q) < \delta$. In fact, we can assume $Q$ is a simple polygonal chain, as Ribó Mor's results \cite{r-rcpps-06} extend to polygonal chains.  Fix a positive real number $\e\ll\delta$ such that the intersection of $Q$ with any closed disk of radius~$\e$ centered on a point of $Q$ is connected; it suffices for $\e$ be less than the \emph{minimum local feature size} of~$Q$~\cite{r-nsaq2-95}.  Let $\hat{Q}$ denote the boundary of the $\e$-neighborhood of $Q$.  Then $\hat{Q}$ is a simple closed curve where $\Frechet(\hat{P}, \hat{Q}) < \delta + 2\e$, which implies that $\hat{P}$ is weakly simple.

On the other hand, suppose the polygon $\hat{P}$ is a weakly simple. 
Then for any $\delta >0$ there is a simple polygon $\hat{Q}$ such that $\Vdist(\hat{P}, \hat{Q}) < \delta$.  
Let $Q$ be either subpath of $\hat{Q}$ between the vertex corresponding to $p_0$ to the vertex corresponding to $p_n$.  
Then $Q$ is a simple polygonal chain with $\Vdist(P, Q) < \delta$, which implies that $P$ is weakly simple.
\end{proof}

\subsection{Graph Drawings}

There is a natural generalization of weak simplicity to arbitrary graph drawings.  
Any graph can be regarded as a topological space, specifically, a branched 1-manifold.  
A \EMPH{planar drawing} of a graph $H$ is just a continuous map from $H$ to the plane; a drawing is \EMPH{simple} or an \EMPH{embedding} if it is injective.  
The Fréchet distance between two planar drawings $P$ and $Q$ of the same graph $H$ is naturally defined as
\[
	\Frechet(P, Q) = \inf_{\phi\colon H\to H} \max_{x\in H} d(P(\phi(x)),Q(x))
\]
where the infimum is taken over all automorphisms of $H$ (homeomorphisms from $H$ to itself).
We can define a planar drawing $P$ to be \EMPH{weakly simple} (or a \EMPH{weak embedding}) if, for any $\e > 0$, there is a planar embedding $Q$ of $H$ with $\Frechet(P, Q)<\e$.  
This definition is consistent with our existing definitions of weakly simple closed curves in the plane.   
It is a natural open question whether one can decide in polynomial time whether a given straight-line drawing of a planar graph is a weak embedding; Cortese \etal~\cite{cbpp-ecpg-09} observe that it is not sufficient to check whether every cycle in the drawing is weakly simple.

However, a generalization of our algorithm actually solves this problem when the graph $H$ being drawn is a disjoint union of cycles.  
In fact, the algorithm is unchanged except for the termination condition; when the main loop terminates, the image graph is the union of disjoint cycles and single segments, and~$P$ is weakly simple if and only if each component of~$P$ either traverses some component of the image graph exactly once or maps to an isolated segment.  
Moreover, using the doubling trick described above in Section \ref{S:polygonal-chains}, our algorithm can also be applied to disjoint unions of cycles \emph{and paths}.  
Details will appear in the full version of the paper.

Any algorithm to determine whether a graph drawing is a weak embedding must handle the special case where the image of the drawing is a simple path.  
This special case is equivalent to the \emph{strip planarity} problem recently introduced by Angelini \etal~\cite{addf-spt-13} as a variant of clustered planarity \cite{cbpp-ecpg-09, fce-pcg-95} and level planarity \cite{jlm-lptlt-98}.  
The strip planarity problem is open even when the graph $H$ is a tree.

\subsection{Surface Graphs}

Our algorithm can also be generalized to surfaces of higher genus.  A closed curve $P$ in an arbitrary surface $\Sigma$ is weakly simple if for any $\e > 0$ there is a simple (injective) closed curve $Q$ in the same surface, such that $\Frechet(P,Q) < \e$, where Fréchet distance is defined with respect to an arbitrary metric on~$\Sigma$.
\begin{theorem}
Given a closed walk $P$ of length $n$ in an arbitrary surface-embedded graph, we can determine whether $P$ is weakly simple in $O(n\log n)$ time.  
\end{theorem}

Our algorithms for detecting weakly simple polygons use the geometry of the plane only in the preprocessing phase, where we apply a sweep-line algorithm to remove forks and to construct the image graph.  Cortese \etal\ already describe topological versions of our expansion operations, which use topological disks instead of ellipses, as well as their proofs of correctness \cite{cbpp-ecpg-09}.  The only minor subtlety is the termination condition when the underlying surface is non-orientable.  For any integer $k>0$, let $\gamma^k$ be the cycle that wraps around some simple cycle $\gamma$ exactly $k$ times.  Then $\gamma^k$ is weakly simple if and only if $k = 1$, or $k=2$ and $\gamma$ is orientation-reversing \cite{r-ajccs-62, z-aekf-65, c-wns2-72}.   Again, details will appear in the full version of the paper. 

\subsection{Faster?}

Finally, perhaps the most immediate open question is how to improve the $O(n^2\log n)$ running time of our algorithm for arbitrary polygons with both spurs and forks.  A direct generalization of bar expansion seems unlikely; Lemma \ref{L:bar} does not generalize to polygons with spurs.  Nevertheless, we conjecture that the quadratic blowup from subdividing edges at forks can be avoided, and that the running time can be improved to $O(n\log n)$.
%

\end{document}